\documentclass[11pt]{article}
\usepackage[utf8]{inputenc}
\usepackage{amsfonts,latexsym,amsthm,amssymb,amsmath,amscd,euscript,mathrsfs,graphicx}
\usepackage{framed}
\usepackage{fullpage}
\usepackage{color}
\usepackage{enumitem}
\usepackage[obeyFinal,textsize=scriptsize,shadow]{todonotes}
\usepackage{hyperref}
\usepackage{authblk}

\usepackage[utf8]{inputenc}
\usepackage{graphicx}

\newtheorem{theorem}{Theorem}[section]
\newtheorem{proposition}[theorem]{Proposition}
\newtheorem{lemma}[theorem]{Lemma}
\newtheorem{corollary}[theorem]{Corollary}
\theoremstyle{definition}
\newtheorem{defn}[theorem]{Definition}

\theoremstyle{remark}
\newtheorem*{remark}{Remark}

\newcommand{\BC}{\mathbb C}
\newcommand{\BE}{\mathbb E}

\newcommand{\BN}{\mathbb N}
\newcommand{\BP}{\mathbb P}

\newcommand{\BR}{\mathbb R}
\newcommand{\BZ}{\mathbb Z}
\newcommand{\eps}{\varepsilon}
\renewcommand{\arg}{\text{arg }}
\renewcommand{\Re}{\text{Re }}
\renewcommand{\Im}{\text{Im }}

\title{Improved Algorithms for Population Recovery from the Deletion Channel}
\author{Shyam Narayanan\thanks{Email: \texttt{shyamsn@mit.edu}. Research supported by the MIT Akamai Fellowship.}}
\affil{Massachusetts Institute of Technology}
\date{\today}

\begin{document}

\maketitle

\begin{abstract}
    The \emph{population recovery} problem asks one to recover an unknown distribution over $n$-bit strings given access to independent noisy samples of strings drawn from the distribution. Recently, Ban et al. \cite{BanCFSS19} studied the problem where the noise is induced through the deletion channel. This problem generalizes the famous \emph{trace reconstruction} problem, where one wishes to learn a single string under the deletion channel.

    Ban et al. showed how to learn $\ell$-sparse distributions over strings using $\exp\big(n^{1/2} \cdot (\log n)^{O(\ell)}\big)$ samples. In this work, we learn the distribution using only $\exp\big(\tilde{O}(n^{1/3}) \cdot \ell^2\big)$ samples, by developing a higher-moment analog of the algorithms of \cite{DeOS17, NazarovP17}, which solve trace reconstruction in $\exp\big(\tilde{O}(n^{1/3})\big)$ samples. We also give the first algorithm with a \emph{runtime} subexponential in $n$, solving population recovery in $\exp\big(\tilde{O}(n^{1/3}) \cdot \ell^3\big)$ samples and time.

    Notably, our dependence on $n$ nearly matches the upper bound of \cite{DeOS17, NazarovP17} when $\ell = O(1)$, and we reduce the dependence on $\ell$ from doubly to singly exponential. Therefore, we are able to learn large mixtures of strings: while Ban et al.'s algorithm can only learn a mixture of $O(\log n/\log \log n)$ strings with a subexponential number of samples, we are able to learn a mixture of $n^{o(1)}$ strings in $\exp\big(n^{1/3 + o(1)}\big)$ samples and time.
\end{abstract}

\newpage

\section{Introduction} \label{Introduction}

\emph{Population Recovery} is an unsupervised learning problem that has recently become of great interest in theoretical computer science \cite{DvirRWY12, MoitraS13, BatmanIMP13, LovettZ15, DeST16, WigdersonY16, PolyanskiySW17, DeOS17b, BanCFSS19, BanCSS19}. In population recovery, there exists some unknown distribution $\mathcal{D}$ over some set (or ``population'') of strings in $\{0, 1\}^n$, which must be learned accurately through queries. For each query, a string $x = x_1 x_2 \cdots x_n$ is drawn from the distribution $\mathcal{D}$, and some noisy version of $x$ is returned, which notably makes determining $\mathcal{D}$ substantially more difficult. The primary noisy versions that have been studied are the \emph{bit-flip} noise model, where each $x_i$ is replaced with $1-x_i$ independently with some probability $q < \frac{1}{2}$, and the \emph{erasure} noise model, where each $x_i$ is independently replaced with some other symbol, such as `?', with some probability $q < 1.$

Recently, Ban et al. \cite{BanCFSS19} considered a version of population recovery under the \emph{deletion channel} noise model. In this model, the noisy version of $x$ is created by removing each $x_i$ independently with probability $0 < q < 1$, but rather than replacing $x_i$ with `?', $x_i$ is simply deleted and the remaining bits are concatenated together. As an example, if $x$ were the string $11001110,$ and if a query of $x$ from the erasure noise model returned $?1?0???0$, the deletion channel would instead return $100,$ without any $?$ symbols. Since the deletion channel throws away information about the initial coordinate of each of the bits in $x$, learning $\mathcal{D}$ in the deletion channel model is noticeably harder than learning $\mathcal{D}$ in the erasure model.

In this work, like in \cite{BanCFSS19}, we consider $\mathcal{D}$ to be a sparse distribution. In other words, there is some integer $\ell$ such that we are promised that $\mathcal{D}$ is a distribution over at most $\ell$ strings. This matches the variant of population recovery studied in \cite{LovettZ15, WigdersonY16, DeST16, BanCFSS19, BanCSS19}, though notably only \cite{BanCFSS19, BanCSS19} considers the deletion channel as opposed to bit-flip or erasure. We now formally define the problem we study in this paper.

\textbf{Population Recovery Problem:} Let $n, \ell$ be positive integers, and let $0 < \eps, p < 1$ be real numbers. Suppose that $\mathcal{D}$ is an unknown distribution over at most $\ell$ unknown strings $x_1, x_2, \dots, x_\ell \in \{0, 1\}^{n}$. Suppose that we are given access to $K$ independent traces, where each trace is formed by first choosing $x$ from the distribution $\mathcal{D}$ over $\{x_1, \dots, x_\ell\}$, and then returning $\tilde{x}$, which is a noisy copy of $x$ in the deletion channel, where each bit of $x$ is independently retained with probability $p$ (i.e., deleted with probability $q = 1-p$). Then, how small can $K$ be so that there is an algorithm that, given the $K$ traces, can reconstruct some distribution $\mathcal{D}'$ so that the total variation distance $d_{\text{TV}}(\mathcal{D}, \mathcal{D}') \le \eps$ with probability at least $2/3$?

Even in the case where $\ell = 1,$ this problem is notably difficult, and is equivalent to the \emph{trace reconstruction problem}, which has been studied extensively over the past two decades \cite{Levenshtein01a, Levenshtein01b, BatuKKM04, KannanM05, HolensteinMPW08, ViswanathanS08, McGregorPV14, DeOS17, NazarovP17, PeresZ17, HoldenPP18, HartungHP18, HoldenL18, Chase19, ChenDLSS20, Chase20}. The trace reconstruction problem asks for the minimum number of traces $K$ such that if $x$ is a single unknown string in $\{0, 1\}^n$, an algorithm can recover $x$ seeing only $K$ independent random traces of $x$, where each trace is a noisy copy of $x$ in the deletion channel, i.e., each trace is formed by independently retaining each bit of $x$ with probability $p$. While this problem is of much active interest, even in this case the bounds are very poorly understood. The state-of-the-art upper bound on $K$ is $\exp\left(\tilde{O}(n^{1/5})\right)$ \cite{Chase20} (which very recently improved on $K = \exp\left(O(n^{1/3})\right)$ \cite{DeOS17, NazarovP17}) for $0 < p < 1$, whereas the state-of-the-art lower bound is significantly smaller, at $\Omega(n^{3/2}/\log^{16} n)$ for constant $0 < p < 1$ \cite{Chase19}. Consequently, one cannot prove polynomial upper bounds or even an $\exp\left(o(n^{1/5})\right)$ upper bound for the population recovery problem under the deletion channel without also improving the upper bound for trace reconstruction. Given the difficulty of trace reconstruction, population recovery may seem almost impossible. It appears that ``this is getting out of hand - now, there are $\ell$ of them!'' But as we will prove, our algorithms can recover a mixture of up to $n^{\eps}$ strings for any $\eps = o(1)$ without either the sample complexity or the runtime suffering by much. 

We remark that apart from population recovery, many variants of trace reconstruction have also recently been studied. These include coded trace reconstruction \cite{CheraghchiGMR19, BrakensiekLS19}, trace reconstruction under cyclic shifts \cite{NR20}, trace reconstruction over matrices \cite{KrishnamurthyMMP19} and over trees \cite{DaviesRR19}, and trace reconstruction under restricted hypotheses on the initial strings \cite{KrishnamurthyMMP19}.

\subsection{Our Results}

As before, let $n$ be the length of the unknown string, $\ell$ be the population size, $p$ be the retention probability (so $q = 1-p$ is the deletion probability), and $\eps$ be the allowed error in total variation distance. Ban et al. \cite{BanCFSS19} was the first paper to study the population recovery problem from the deletion channel. We begin by stating their results.

\begin{theorem} \cite{BanCFSS19} \label{old}
    For parameters $n, \ell, p, \eps,$ there exists an algorithm that can solve the population recovery problem with
\[K = \frac{1}{\eps^2} \cdot \left(\frac{2}{p}\right)^{\sqrt{n} \cdot (\log n)^{O(\ell)}}\]
    samples. However, for any constant $0 < p < 1$ and $\eps = 0.49,$ and for $\ell \le \sqrt{n},$ any algorithm must use at least
\[K = \frac{\Omega(n/\ell^2)^{(\ell+1)/2}}{\ell^{3/2}}\]
    samples. If $\ell \le n^{0.499},$ note that this means any algorithm must use at least $n^{\Omega(\ell)}$ samples.
\end{theorem}

In this paper, we strongly improve the upper bound of Theorem \ref{old} with our main theorem. 

\begin{theorem} \label{main}
    For parameters $n, \ell, p, \eps,$ there exists an algorithm that can solve the population recovery problem with
\begin{equation} \label{MainEq1}
K = \eps^{-2} \cdot \exp\left(O\left(n^{1/3} (\log n)^{2/3} \ell^{2} p^{-2/3}\right)\right)
\end{equation}
    samples. There also exists an algorithm that can solve the population recovery problem with
\begin{equation} \label{MainEq2}
    K = \eps^{-2} \cdot \exp\left(O\left(n^{1/3} (\log n)^{2/3} \ell^{7/3} p^{-1/3} + (\log n) \ell^3 p^{-1} \right)\right)
\end{equation}
    samples.
\end{theorem}

Note that Equation \eqref{MainEq1} gives a better upper bound when $p^{-1} < \ell$ or $\frac{n}{\log n} < p^{-1} \cdot \ell^3,$ whereas Equation \eqref{MainEq2} gives a better bound otherwise.

One issue with the algorithms for both Theorem \ref{old} and Theorem \ref{main} is that both algorithms have runtime exponential in $n$. This is because both algorithms, while having query complexity subexponential in $n$, essentially check all possible $\ell$-sparse distributions $\mathcal{D}$ to find a suitable match. Therefore, a natural question is whether one can achieve a faster algorithm as well. In this paper, we provide the first algorithm for population recovery under the deletion channel with runtime subexponential in $n$. Namely, we prove the following theorem:

\begin{theorem} \label{MainAlg}
     Let parameters $n, \ell, p, \eps$ be as usual. Suppose that $0 < \alpha \le 1$ is some parameter so that for all $x$ in the support of our unknown distribution $\mathcal{D},$ $\BP_{y \sim \mathcal{D}}(y = x) \ge \alpha$. Then, there exists an algorithm that can solve the population recovery problem with samples and runtime bounded by
\begin{equation} \label{MainEq3}
    (\eps^{-2} + \alpha^{-2} \log \alpha^{-1}) \exp\left(O\left(n^{1/3} (\log n)^{2/3} \ell^{2} p^{-2/3}\right)\right).
\end{equation}
    This algorithm can also be modified so that even if $\alpha$ is unknown and possibly arbitrarily small, the algorithm can solve the population recovery problem with samples and runtime bounded by
\begin{equation} \label{MainEq4}
    \eps^{-2} \log \eps^{-1} \exp\left(O\left(n^{1/3} (\log n)^{2/3} \ell^{3} p^{-2/3}\right)\right).
\end{equation}
\end{theorem}

\begin{remark}
    We note that the bounds of Theorem \ref{MainAlg} correspond to Equation \eqref{MainEq1} of Theorem \ref{main}: we either replace the multiplicative factor of $\eps^{-2}$ with $\eps^{-2} + \alpha^{-2} \log \alpha^{-1}$ (as in Equation \eqref{MainEq3}) or both a multiplicative factor of $\log \eps^{-1}$ and an additional factor of $\ell$ in the exponent (as in Equation \eqref{MainEq4}). One can also produce analogous bounds that correspond to Equation \eqref{MainEq2}, though the proof is almost identical, so in this paper we do not prove the bounds corresponding to Equation \eqref{MainEq2}.
\end{remark}

Our results have three significant improvements over the previous upper bound of Ban et al. First, we reduce the exponent's dependence on $n$ from $\tilde{O}(n^{1/2})$ to $\tilde{O}(n^{1/3})$. Second, the dependence on $\ell$ is reduced from doubly exponential to nearly singly exponential, which is more in line with the lower bound in Theorem \ref{old}. Note that the upper bound of Theorem \ref{old} shows that for $\ell = o\left(\frac{\log n}{\log \log n}\right),$ $p$ constant, and $\eps^{-1}$ polynomially bounded, only $\exp\left(n^{1/2 + o(1)}\right)$ queries are needed. However, both Theorems \ref{main} and \ref{MainAlg} allow the mixture to have up to $\ell = n^{o(1)}$ distinct strings and still only $\exp\left(n^{1/3+o(1)}\right)$ queries are needed for constant $p$, which is a major improvement over Theorem \ref{old}. The final improvement is that we now also have a much faster algorithm, which has runtime exponential in $n^{1/3}$ rather than exponential in $n$ as in \cite{BanCFSS19}.

We note that our new bounds may be weaker than the previous bounds when $p = o(n^{-1/2}),$ i.e. when the deletion probability is sufficiently close to $1$. However, we show a simple reduction to the $p = n^{-1/2}$ case, and as a result prove the following theorem as well.
    
\begin{theorem} \label{Smallp}
    For parameters $n, \ell, p, \eps,$ if $p \le \frac{1}{2} \cdot n^{-1/2},$ there exists an algorithm that can solve the population recovery problem with
\[K = \eps^{-O(\log p^{-1})} \cdot \exp\left(O\left(\sqrt{n} \log n \cdot \ell^3 \cdot \log p^{-1}\right)\right)\]
    samples.
\end{theorem}

Theorem \ref{Smallp} gets matching bounds with respect to $n$ and $p$ for $p \le n^{-1/2},$ but with significantly better bounds with respect to $\ell$. Thus, our upper bounds are stronger than Theorem \ref{old} for all parameter regimes, except when both $p \le \frac{1}{2} \cdot n^{-1/2}$ and $\eps \le e^{-\ell^3 \cdot \sqrt{n} \log n}$.

\subsection{Comparison to Trace Reconstruction and Other Population Recovery Models}

As we noted before, our bounds for constant $\ell$ nearly match the known bounds of \cite{DeOS17, NazarovP17} when $\ell = 1,$ which require $\exp\left(O(p^{-1/3}n^{1/3})\right)$ samples when $p \ge n^{-1/2}$, so we lose only a $(\log n)^{2/3}$ factor in our exponent. 
After this paper was initially released, however, Chase \cite{Chase20} improved the sample complexity of trace reconstruction to $\exp\left(\tilde{O}(n^{1/5})\right)$, as stated above.
Since the bounds of \cite{DeOS17, NazarovP17} were the best known for trace reconstruction at the time of this paper's (as well as Ban et al.'s \cite{BanCFSS19}) initial release, a natural open question is whether our bounds for population recovery (for both query and time) can be improved to $\exp\left(\tilde{O}(n^{1/5} \cdot \text{poly}(\ell))\right)$.

Finally, we compare our results to population recovery under the erasure and bit-flip noise models. We remark that population recovery for an $\ell$-sparse distribution under these two models has significantly better upper bounds for constant $0 < p < 1$. Indeed, in both the bit-flip and the erasure model, there exist algorithms that can reconstruct $\mathcal{D}$ up to total variation distance $\eps$ in samples polynomial in $n, \ell,$ and $1/\eps$ \cite{MoitraS13, DeST16}. However, we note that bounds like these cannot even be possible in the deletion channel model, since if $\eps = 0.49$ and $\ell = n^{0.49}$, any algorithm must take $n^{\Omega(\ell)}$ queries by Theorem \ref{old}. Moreover, note that in the case of $\ell = 1$ for the bit-flip or erasure noise models, reconstruction is easily achievable in $O(\log n)$ traces, whereas we cannot hope for anything better than $\tilde{\Omega}(n^{3/2})$ in the deletion channel noise model, due to the trace reconstruction lower bound of \cite{Chase19}.

\subsection{Outline of Paper}

We briefly outline the rest of this paper. In Section \ref{Outline}, we explain the methods used in proving our main theorems, and compare our techniques to those of previous papers. In Section \ref{Preliminaries}, we explain some useful preliminary results. In Section \ref{PopulationRecoveryBounds}, we prove Theorem \ref{main}. In Section \ref{AlgSection}, we prove Theorem \ref{MainAlg}. Finally, in Section \ref{SmallpSection}, we prove Theorem \ref{Smallp} by a simple modification of Theorem \ref{main}.

\section{Proof Outline} \label{Outline}

Our approach can be viewed as a higher-moment generalization of the technique of \cite{DeOS17, NazarovP17}, which proved trace reconstruction was doable in $\exp\left(O\left(n^{1/3}\right)\right)$ samples. First, we explain the ideas used in previous papers. We then outline the ideas that we develop to prove Theorems \ref{main} and \ref{MainAlg}.

\subsection{Techniques of Previous Papers}

The proofs in \cite{DeOS17, NazarovP17} are very similar, but we will explain the idea of Nazarov and Peres \cite{NazarovP17} as it is slightly simpler. We will also show only how an algorithm can distinguish between two distinct strings $x, x'$ for any $x \neq x'$ rather than reconstruct $x$, as there exists a simple reduction from the former to the latter. The main idea of Nazarov and Peres was to choose a complex number $z$, and for the unknown string $x$, construct an unbiased estimator for $P(z; x):= \sum_{i = 1}^{n} x_i z^i,$ which is a degree $n$ polynomial in $z$, using only the trace $\tilde{x}$. They showed, using a simple combinatorial argument, that if one was given a trace $\tilde{x} = \tilde{x}_1 \cdots \tilde{x}_n$ of $x$, where we have padded the trace with additional $0$'s, then $\sum \tilde{x}_i w^i$ is an unbiased estimator of $P(z; x)$ for $w = \frac{z-q}{p},$ up to a scaling factor. Recall that $q$ is the deletion probability and $p = 1-q$ is the retention probability. They then used a result of Peter Borwein and Tam\'{a}s Erd\'{e}lyi \cite{BorweinE97}, which proves that for any strings $x \neq x' \in \{0, 1\}^n$ and any $\eps < 1,$ there is some complex number $z$ with magnitude $1$ and with argument at most $O(\eps)$ such that $|P(z; x)-P(z; x')| \ge \exp\left(-\eps^{-1}\right).$ Thus, the unbiased estimators of $P(z; x)$ can be used to successfully distinguish between $x$ and $x',$ assuming that we can bound the variance of $\sum \tilde{x}_i w^i.$ However, assuming that $p$ is a constant between $0$ and $1$, one can easily verify that if $|z| = 1$ and $|\arg z| = O(\eps),$ then $w = 1+O(\eps^2),$ which means $|\sum \tilde{x}_i w^i|$ is uniformly bounded by $e^{O(\eps^2 \cdot n)}.$ Choosing $\eps = n^{-1/3}$ and using the Chebyshev inequality bound shows that $\exp\left(O\left(n^{1/3}\right)\right)$ samples of our unbiased estimator is sufficient to distinguish $P(z; x)$ from $P(z; x'),$ and thus $x$ from $x'.$

The above method is known as a mean-based algorithm. This is because the algorithm above only uses the sample means of $\sum \tilde{x}_i w^i,$ and therefore the algorithms of \cite{DeOS17, NazarovP17} reconstruct $x$ as a function of only the \emph{sample means} of $\tilde{x}_i$ for each $i$. Ban et al. \cite{BanCFSS19} note that a mean based algorithm cannot work for population recovery, even in the case $\ell = 2.$ For instance, if one considers $\mathcal{D}_0$ as a uniform mixture of the strings $0^n = 00 \cdots 0$ and $1^n = 11 \cdots 1,$ and $\mathcal{D}_1$ as a uniform mixture of the strings $0^{n/2} 1^{n/2}$ and $1^{n/2} 0^{n/2},$ then $\BE_{x \sim \mathcal{D}_0}[\tilde{x}_i] = \BE_{x \sim \mathcal{D}_1}[\tilde{x}_i]$ for all $i$. In other words, for all $i \le n,$ the expectation of $\tilde{x}_i$ is the same regardless of whether $x$ is drawn from $\mathcal{D}_0$ or from $\mathcal{D}_1.$ To avoid this issue, Ban et al. uses a very different approach, based on what is called the the \emph{$k$-deck} of $\tilde{x},$ which counts the number of times that each string $s$ of length $k$ appears as a subsequence of $\tilde{x}$. They need $k$ to be approximately $\sqrt{n} \cdot (\log n)^{O(\ell)},$ and need exponential in $k$ traces to get a sufficiently good approximation of the $k$-deck of the original distribution, which counts the expected number of times that each $s$ appears as a subsequence of $x$ drawn from $\mathcal{D}$. This allows them to get a sample complexity bound of approximately $\exp\left(\sqrt{n} \cdot (\log n)^{O(\ell)}\right)$.

Finally, for the trace reconstruction problem, \cite{DeOS17} briefly noted that one could use the ideas of \cite{HolensteinMPW08} to reconstruct the original string using both $\exp\left(O\left(n^{1/3}\right)\right)$ queries and time. The idea in \cite{HolensteinMPW08} for a fast algorithm was to use linear programming to reconstruct the bits of the original string $x$ one at a time. The natural problem we are attempting to solve is an integer linear program, as $x \in \{0, 1\}^n$ and $P(z; x) = \sum x_i z^i$ is a linear function of $x$. However, if we know the values of $x_1, \dots, x_{i-1}$ and are trying to determine $x_i,$ it turns out that we can solve for $x_i \in \{0, 1\}$ even if we relax $x_{i+1}, \dots, x_n$ to be in the interval $[0, 1]$, so one can use linear programming to solve for $x$ efficiently. This idea will be a useful step in proving Theorem \ref{MainAlg}.

\subsection{Our Techniques: Theorem \ref{main}}

The ideas of \cite{DeOS17, NazarovP17} serve as our starting point for tackling the population recovery problem. However, to overcome the limitations of mean-based algorithms, we utilize higher moments. A simple motivation for this is: in the example of $\mathcal{D}_0$ a mixture of $0^n$ and $1^n$ and $\mathcal{D}_1$ a mixture of $0^{n/2}1^{n/2}$ and $1^{n/2} 0^{n/2},$ there is more covariance between $\tilde{x}_i$ and $\tilde{x}_j$ for any $i \neq j$ if $x$ were drawn from $\mathcal{D}_0$ than if $x$ were drawn from $\mathcal{D}_1.$ Given a random trace $\tilde{x}$ of some string $x$, our first step will be to construct an unbiased estimator of $P(z; x)^k$ only based on $\tilde{x}$ instead of $x$, rather than just an unbiased estimator of $P(z; x)$. (Recall that $P(z; x) := \sum_{i = 1}^{n} x_i z^i$.) As a result, if we have an $\ell$-sparse distribution $\mathcal{D}$ over $\{0, 1\}^n$ that equals $x^{(i)} \in \{0, 1\}^n$ with probability $a_i$ for $1 \le i \le \ell,$ our estimator will have expectation $\sum_{i = 1}^{\ell} a_i P(z; x^{(i)})^k$.

Even after noting this, we are still left with two major technical challenges that are more difficult in the case of general $k$ than in the $k = 1$ case. First, we must successfully produce an unbiased estimator of $P(z; x)^k$ using just random traces of $x$. Second, we must prove that as long as $\mathcal{D}_0$ and $\mathcal{D}_1$ differ sufficiently, the averages of the $P(x; z)^k$'s (weighted according to the mixture weights) also differ significantly.

How do we actually construct an unbiased estimator of $P(z; x)^k$? Doing so is more difficult for general $k$ than $k = 1,$ where we can use $\sum \tilde{x}_i w^i$ for some $w \in \BC.$ For general $k$, the rough idea is to consider complex numbers $w_1, w_2, \dots, w_k \in \BC$ and look at $\sum_{1 \le i_1 < \dots < i_k \le n} \tilde{x}_{i_1} \cdots \tilde{x}_{i_k} w_1^{i_1} w_2^{i_2} \cdots w_k^{i_k}.$ One can provide an explicit formula for the expectation of this sum, and if $w_1, \dots, w_k$ are carefully chosen in terms of $z,$ one can get an unbiased estimator for $\sum_{1 \le i_1 < \dots < i_k \le n} x_{i_1} \cdots x_{i_k} z^{i_1+\dots+i_k},$ which looks very similar to the expansion of $P(z; x)^k,$ with the exception of ignoring all terms where some of the $x_i$'s are equal. However, the remaining terms in $P(z; x)^k$ can be grouped into sums of the form $\sum_{1 \le i_1 < \cdots < i_{k'} \le n} x_{i_1} \cdots x_{i_{k'}} z^{b_1 i_1 + \cdots + b_{k'} i_{k'}}$ for some $k' \le k$ and some positive integers $b_1, \dots, b_{k'}$ that add to $k$. We can get unbiased estimators for these terms from $\sum_{1 \le i_1 < \dots < i_{k'} \le n} \tilde{x}_{i_1} \cdots \tilde{x}_{i_k} w_1^{i_1} w_2^{i_2} \cdots w_{k'}^{i_{k'}}$ for possibly varying choices of $w_1, w_2, \dots, w_{k'},$ but overall by adding all of these sums, we can get an unbiased estimator for $P(z; x)^k.$

Our unbiased estimators give us good approximations for $\sum a_i P(z; x^{(i)})^k$ - we will construct such estimators for all integers $k$ between $0$ and $2 \ell.$ Now, let $\mathcal{D}_0$ be a distribution that equals $x^{(i)}$ with probability $a_i$ for some strings $x^{(i)}$ and probabilities $a_i,$ and let $\mathcal{D}_1$ be a distribution that equals $y^{(i)}$ with probability $b_i$ for some strings $y^{(i)}$ and probabilities $b_i$. The goal is to show that if $\mathcal{D}_0$ and $\mathcal{D}_1$ are $\ell$-sparse distributions over $\{0, 1\}^n$ with total variation distance at least $\eps$,
\begin{equation} \label{Quantity}
    \left|\sum_{i = 1}^{\ell} a_i P(z; x^{(i)})^k - \sum_{i = 1}^{\ell} b_i P(z; y^{(i)})^k\right|
\end{equation}
is not too small for some $k \le 2 \ell$ and some $z \in \BC$. This suffices to distinguish between distributions $\mathcal{D}_0$ and $\mathcal{D}_1,$ since we can average the unbiased estimators of $\sum a_i P(z; x^{(i)})^k$ over sufficiently many samples to successfully distinguish $\mathcal{D}_0$ from $\mathcal{D}_1.$ To actually prove this, we reduce this problem to providing a lower bound on $\prod_{x \neq x'} |P(z; x) - P(z; x')|$ for all pairs of distinct strings $x, x' \in \{0, 1\}^n$ such that $x, x' \in \{x^{(1)}, \dots, x^{(\ell)}, y^{(1)}, \dots, y^{(\ell)}\}.$ We do this by constructing a particular \emph{Vandermonde matrix} using the complex numbers $P(z; x^{(i)})$ and $P(z; y^{(i)})$ that has determinant $\prod_{x \neq x'} (P(z; x) - P(z; x')),$ which is also a polynomial in $z$. We finally use the results from \cite{BorweinE97} to provide a lower bound on the magnitude of this polynomial for some $z$ with magnitude $1$ and argument approximately bounded by $n^{-1/3}.$ We can use this to provide a lower bound on the least singular value of the Vandermonde matrix, which will give us a lower bound for the quantity in Equation \eqref{Quantity}. This will turn out to be sufficient for distinguishing between any $\mathcal{D}_0$ and $\mathcal{D}_1$ with total variation distance at least $\eps$.

\subsection{Our Techniques: Theorem \ref{MainAlg}}

In the trace reconstruction problem, we can recover the polynomial $P(z; x)$ by recovering the coefficients one at a time using linear programming. A natural follow-up idea is to recover the polynomials $\sum_{i = 1}^{\ell} a_i P(z; x^{(i)})^k$ for all $k \le O(\ell)$ by determining the coefficients one at a time. However, this will turn out to be problematic, as in the trace reconstruction algorithm, recovering a coefficient of $P(z; x)$ means we just have to determine if the coefficient is $1$ or $0$. Thus, if we can approximate the degree $i$ coefficient of $P(z; x)$, we can determine $x_i$ exactly. However, when the $a_i$'s are arbitrary reals, we can only approximately determine coefficients of $\sum a_i P(z; x^{(i)})$, so if our guess for the degree $1$ coefficient is slightly off, it may cause our guesses for the degree $i$ coefficients to be far off for larger values of $i$.

To fix this issue, our main insight is to not attempt to solve for the polynomials $\sum_{i = 1}^{\ell} a_i P(z; x^{(i)})^k$ but instead solve for the elementary symmetric polynomials of $P(z; x^{(1)}), \dots, P(z; x^{(\ell)}).$ The motivation for this is that these polynomials, in terms of $z$, have integer coefficients, so approximately determining their coefficients is just as good as exactly determining their coefficients. Given exact values for $\sum a_i P(z; x^{(i)})^k$, one can reconstruct the $j$th elementary symmetric polynomial,
\[\sigma_j\left(P(z; x^{(1)}), \dots, P(z; x^{(\ell)})\right) := \sum\limits_{1 \le i_1 < i_2 < \cdots < i_j \le \ell} P(z; x^{(i_1)}) \cdots P(z; x^{(i_j)}),\]
based on a linear algebraic technique called Prony's method. However, we will only receive good estimates for $\sum a_i P(z; x^{(i)})^k$ rather than their exact values. One can likewise return estimates for $\sigma_j\left(P(z; x^{(1)}), \dots, P(z; x^{(\ell)})\right)$, assuming our estimates for $\sum a_i P(z; x^{(i)})^k$ are sufficiently accurate, though determining how good our estimates for $\sum a_i P(z; x^{(i)})^k$ need to be will require significant linear algebra. In other words, we must show that recovering the elementary symmetric polynomials from the values $\sum a_i P(z; x^{(i)})^k$ is robust to error in our estimates for $\sum a_i P(z; x^{(i)})^k$. Proving this will turn out to be a key technical challenge.

Once we have approximately established the values $\sigma_j\left(P(z; x^{(1)}), \dots, P(z; x^{(\ell)})\right)$ for several values of $z$ and for all $1 \le j \le \ell$, we can reconstruct the actual coefficients of the polynomials $\sigma_j\left(P(z; x^{(1)}), \dots, P(z; x^{(\ell)})\right)$. We use a similar method to that of \cite{HolensteinMPW08}, by using linear programming to construct the coefficients one at a time. We know that the coefficients are all nonnegative integers, and it is easy to see that all coefficients are bounded by $n^{O(\ell)}.$ Thus, if we are trying to solve for some polynomial $\sigma_j\left(P(z; x^{(1)}), \dots, P(z; x^{(\ell)})\right) =: \sum t_i z^i$, and we know the coefficients $t_1, \dots, t_{i-1},$ we will set $t_i$ to be a nonnegative integer at most $n^{O(\ell)}$ and relax the constraints for $t_j$ to be in $[0, n^{O(\ell)})]$ for all $j > i$. Then, using linear programming, we can solve for $t_i,$ and by induction return all coefficients of the polynomial.

Once we have determined the coefficients of the polynomials $\sigma_j\left(P(z; x^{(1)}), \dots, P(z; x^{(\ell)})\right)$ for all $j$, we can use standard factoring methods to recover the original polynomials $P(z; x^{(i)})$ for all $1 \le i \le \ell.$ It remains to determine the values $a_1, \dots, a_\ell,$ where we recall that $a_i = \BP(x = x^{(i)})$ when $x$ is drawn from our distribution $\mathcal{D}$. This, however, will turn out to be quite simple, as we can use linear programming to solve for the $a_i$'s.

\section{Preliminaries} \label{Preliminaries}

We will need some simple results about complex numbers, as well as a ``Littlewood-type'' result about bounds on polynomials on arcs of the unit circle \cite{BorweinE97}. We note that the latter result requires complex analysis, though we will not have to use any knowledge of complex analysis besides this result as a black box. Finally, we will need two matrix formulas: the Sherman-Morrison-Woodbury matrix identity and Weyl's matrix inequality for non-Hermitian matrices.

First, we explain a basic definition we will use involving complex numbers.

\begin{defn}
    For $z \in \BC,$ let $|z|$ be the \emph{magnitude} of $z$, and if $z \neq 0$, let $\arg z$ be the \emph{argument} of $z$, which is the value of $\theta \in (-\pi, \pi]$ such that $\frac{z}{|z|} = e^{i \theta}$.
\end{defn}

Next, we prove three simple propositions about complex numbers.

\begin{proposition} \label{Complex1}
    Suppose that $0 < p, q < 1$ are real numbers with $p+q < 1.$ Then, if $z, w \in \BC$ such that $|z - (1-p)| \le p$ and $|w-(1-q)| \le q,$ then $|zw - (1-p-q)| \le p+q.$
\end{proposition}

\begin{proof}
    Write $z = 1-p+s$ and $w = 1-q+t$ for $|s| \le p$ and $|t| \le q.$ Then, $zw$ can be expanded as $1-p-q + (pq + t(1-p) + s(1-q) + st).$ However,
\[|pq + t(1-p) + s(1-q) + st| \le pq + |t| \cdot (1-p) + |s| \cdot (1-q) + |s| \cdot |t| \le pq + q(1-p) + p(1-q) + pq = p+q.\]
    This completes the proof.
\end{proof}

\begin{proposition} \label{Complex2}
    Let $0 < p < 1$ be a real number, and let $k$ be a positive integer. Then, if $z$ is a complex number satisfying $\left|z - (1 - \frac{p}{k})\right| \le \frac{p}{k},$ then $|z^k - (1 - p)| \le p.$
\end{proposition}

\begin{proof}
    Equivalently, it suffices to show that if $0 < p < \frac{1}{k},$ then if $|z - (1-p)| \le p,$ then $|z^k - (1-pk)| \le pk.$ But this follows immediately from Proposition \ref{Complex1}, by inducting on $k$.
\end{proof}

\begin{proposition} \label{Complex3}
    Let $z$ be a complex number with $|z| = 1$ and $|\arg z| \le \theta.$ Then, for any $0 < p < 1,$ $\left|\frac{z-(1-p)}{p}\right| \le 1 + \frac{\theta^2}{p^2}.$
\end{proposition}
\begin{proof}
    Write $\frac{z-(1-p)}{p}$ as $1 + \frac{z-1}{p}.$ Then, $\left|1 + \frac{z-1}{p}\right|^2 = 1 + \frac{|z-1|^2}{p^2} + 2 \cdot \frac{\Re (z-1)}{p}.$ However, since $z$ is on the unit circle, $\Re (z-1) \le 0$ and $|z-1| \le \theta.$ Thus, $\left|1 + \frac{z-1}{p}\right|^2 \le 1 + \frac{\theta^2}{p^2}.$ Therefore, $\left|\frac{z-(1-p)}{p}\right| = \left|1 + \frac{z-1}{p}\right| \le \sqrt{1 + \frac{\theta^2}{p^2}} \le 1 + \frac{\theta^2}{p^2}.$
\end{proof}

Next, we state the Littlewood-type result.

\begin{theorem} \textup{\cite{BorweinE97}} \label{Littlewood}
    Let $f(x) = \sum_{j = 0}^{n} a_j x^j$ be a polynomial of degree $n$ with complex coefficients. Suppose there is some positive real number $M$ such that $|a_0|= 1$ and $|a_j| \le M$ for all $0 \le j \le n$. Let $A$ be a subarc of the unit circle $|z| = 1$ in the complex plane with length $0 < a < 2 \pi$. Then, there exists some absolute constant $c_1 > 0$ such that
\[\sup\limits_{z \in A} |f(z)| \ge \exp\left(\frac{-c_1 (1 + \log M)}{a}\right).\]
\end{theorem}

Next, we state the Sherman-Morrison-Woodbury matrix identity.

\begin{theorem} \label{Woodbury}
    Let $A, B \in \BC^{k \times k}$ be complex-valued square matrices. Then, if $A$ and $A+B$ are invertible, 
\[(A+B)^{-1} = A^{-1} - A^{-1} B (I + A^{-1} B)^{-1} A^{-1}\]
    where $I$ is the $k \times k$ identity matrix.
\end{theorem}

Finally, we state the Weyl matrix inequality.

\begin{theorem} \label{Weyl} \textup{\cite[Exercise 1.3.22]{TaoBook}}
    Let $A, B \in \BC^{k \times k}$ be complex-valued square matrices. Then, if $\sigma_1 \ge \sigma_2 \ge \cdots \ge \sigma_k \ge 0$ are the $k$ singular values of $A$ and $\tau_1 \ge \tau_2 \ge \cdots \ge \tau_k \ge 0$ are the $k$ singular values of $A+B,$ then for all $1 \le i \le k,$ $|\sigma_i-\tau_i| \le \|B\|,$ where $\|B\|$ is the operator norm of $B$.
\end{theorem}

\section{Sample Complexity Bound} \label{PopulationRecoveryBounds}

In this section, we will prove Theorem \ref{main}.

First, consider a single string $x \in \{0, 1\}^n$ that we can query traces from. Let $x = x_1 \cdots x_n$, and let $\tilde{x}_1 \cdots \tilde{x}_n$ be a trace of $x$, where we pad the trace with $0$'s until the string becomes length $n$.

First, we create a function $g_m(\tilde{x}, z),$ that takes as input the trace $\tilde{x},$ some complex number $z$, and some positive integer $m.$ It will also depend on $p$ and $n$, but we treat these as fixed. We will show that this function is an unbiased estimator of a certain polynomial of $z$, depending on $x$ and $m$, and provide uniform bounds on $g_m(\tilde{x}, z)$.

\begin{lemma} \label{CreatingG}
    Fix $n$ as the length of $x$, $p$ as the retention probability, and $q = 1-p$ as the deletion probability. Then, for any integer $m \ge 1,$ there exists some function $g_m(\tilde{x}, z)$ such that for all $z \in \BC,$
\[\BE_{\tilde{x}}[g_m(\tilde{x}, z)] = \left(\sum\limits_{i = 1}^{n} x_i z^i\right)^m,\]
    where the expectation is over traces drawn from $x$. Moreover, for any integer $L \ge 1$ and for all $\tilde{x} \in \{0, 1\}^n$ and all $z$ with magnitude $1$ and argument at most $\frac{2 \pi}{L}$ in absolute value,
\[\left|g_m(\tilde{x}, z)\right| \le (p^{-1} m n)^{O(m)} \cdot e^{O(m^2 n/(p^2 L^2))}.\]
    Also, for all $\tilde{x}$ and all $z$ such that $|z - (1-\frac{p}{m})| \le \frac{p}{m},$ 
\[|g_m(\tilde{x}, z)| \le (p^{-1} m n)^{O(m)}.\]
    Finally, $g_m(\tilde{x}, z)$ can be computed in $n^{O(m)}$ time.
\end{lemma}

\begin{proof}
For some $1 \le k \le m,$ fix some complex numbers $w_1, \dots, w_k$ and consider the random variable
\[f(\tilde{x}, w) := \sum\limits_{1 \le i_1 < i_2 < \dots < i_k \le n} \tilde{x}_{i_1} \cdots \tilde{x}_{i_k} w_1^{i_1} w_2^{i_2-i_1} \cdots w_k^{i_k-i_{k-1}}\]
    for $w = (w_1, \dots, w_k)$, which is a random variable since $\tilde{x}$ is random. Given $(w_1, \dots, w_k)$ and $\tilde{x},$ since there are at most ${n \choose k}$ terms, this can be computed in time $n^{O(k)}.$
    
    We first describe $\BE[f(\tilde{x}, w)]$ and choose appropriate values for $w_1, \dots, w_k$. First, we can rewrite
\[f(\tilde{x}, w) = \sum\limits_{\substack{i_1, \dots, i_k \ge 1\\ i_1 + \dots + i_k \le n}} \tilde{x}_{i_1} \tilde{x}_{i_1+i_2} \cdots \tilde{x}_{i_1+i_2+\dots+i_k} w_1^{i_1} w_2^{i_2} \cdots w_k^{i_k}.\]
    For any $j_1, \dots, j_k$, note that $\tilde{x}_{i_1}$ coming from $x_{j_1}$, $\tilde{x}_{i_1+i_2}$ coming from $x_{j_1+j_2}$, etc. means that $j_1 \ge i_1, j_2 \ge i_2, \dots, j_k \ge i_k$. Moreover, even in this case, this will only happen with probability
\[\prod\limits_{r = 1}^{k} \left(p \cdot {j_r-1 \choose i_r-1} p^{i_r-1} q^{j_r-i_r} \right) = p^{\sum i_r} q^{\sum (j_r-i_r)} \prod\limits_{r = 1}^{k} {j_r - 1 \choose i_r - 1}.\]
    Therefore, we have that
\begin{align*}
\BE[f(\tilde{x}, w)] &= \sum\limits_{\substack{i_1, \dots, i_k \ge 1 \\ j_r \ge i_r \\ j_1 + \dots + j_k \le n}} \prod\limits_{r = 1}^{k} \left({j_r - 1 \choose i_r - 1} p^{i_r} q^{j_r-i_r} x_{j_1+\dots+j_r} w_r^{i_r}\right)\\
&=\sum\limits_{\substack{j_1, \dots, j_k \ge 1 \\ j_1 + \dots + j_k \le n}} \prod\limits_{r = 1}^{k} \left(p w_r x_{j_1+\dots+j_r} \cdot \sum\limits_{i_r = 1}^{j_r} {j_r - 1 \choose i_r - 1} p^{i_r-1} q^{j_r-i_r} w_r^{i_r-1}\right) \\
&=\sum\limits_{\substack{j_1, \dots, j_k \ge 1 \\ j_1 + \dots + j_k \le n}} \prod\limits_{r = 1}^{k} \left(p w_r x_{j_1+\dots+j_r} \cdot (p w_r + q)^{j_r-1} \right) \\
&= p^k \frac{w_1 \cdots w_k}{(pw_1+q) \cdots (pw_k+q)} \cdot \sum\limits_{\substack{j_1, \dots, j_k \ge 1 \\ j_1 + \dots + j_k \le n}} x_{j_1} \cdots x_{j_1+\dots+j_k} (pw_1+q)^{j_1} \cdots (pw_k+q)^{j_k}.
\end{align*}

    Now, fix some $z \in \BC$ and $B = (b_1, \dots, b_k)$ an ordered $k$-tuple of positive integers such that $b_1+\dots+b_k \le m.$ For all $1 \le r \le k$, let $w_{B, r} = \frac{z^{b_r + b_{r+1} + \dots + b_k}-q}{p}$. Then, letting $w_B = (w_{B, 1}, \dots, w_{B, k}),$
\begin{align*}
\BE\left[f(\tilde{x}, w_B)\right] &= p^k \cdot \frac{w_{B, 1} \cdots w_{B, k}}{z^{b_1 + 2b_2 + \dots + kb_k}} \cdot \sum\limits_{\substack{j_1, \dots, j_k \ge 1 \\ j_1 + \dots + j_k \le n}} x_{j_1} \cdots x_{j_1+\dots+j_k} z^{(b_1+\dots+b_k)j_1+(b_2+\dots+b_k)j_2+\dots+b_kj_k} \\
&= p^k \cdot \frac{w_{B, 1} \cdots w_{B, k}}{z^{b_1 + 2b_2 + \dots + kb_k}} \cdot \sum\limits_{1 \le i_1 < i_2 < \dots < i_k \le n} x_{i_1} \cdots x_{i_k} z^{b_1 i_1 + \dots + b_k i_k},
\end{align*}
    where we have written $i_r = j_1 + j_2 + \dots + j_r$ for all $1 \le r \le k$. This implies that
\begin{align*}
\left(\sum\limits_{i = 1}^{n} x_i z^i\right)^{m} &= \sum\limits_{\substack{1 \le k \le m \\ b_1, \dots, b_k \ge 1 \\ b_1 + \dots + b_k = m}} {m \choose b_1, b_2, \dots, b_k} \cdot \sum\limits_{1 \le i_1 < i_2 < \dots < i_k \le n} x_{i_1} \cdots x_{i_k} z^{b_1 i_1 + \dots + b_k i_k} \\
&= \sum\limits_{\substack{B = (b_1, \dots, b_k) \\ b_1 + \dots + b_k = m}} {m \choose b_1, b_2, \dots, b_k} \cdot p^{-k} \cdot \frac{z^{b_1+2b_2+\dots+kb_k}}{w_{B, 1} \cdots w_{B, k}} \cdot \BE\left[f(\tilde{x}, w_B)\right].
\end{align*}
    Above, we used the fact that $x_i = x_i^{b_i}$ as $x \in \{0, 1\}$ and $b_i \in \BN.$ 
    
    Now, let 
\[g_m(\tilde{x}, z) := \sum\limits_{\substack{B = (b_1, \dots, b_k) \\ b_1 + \dots + b_k = m}} {m \choose b_1, b_2, \dots, b_k} \cdot p^{-k} \cdot \frac{z^{b_1+2b_2+\dots+kb_k}}{w_{B, 1} \cdots w_{B, k}} \cdot f(\tilde{x}, w_B).\]
    Since the number of tuples $(b_1, b_2, \dots, b_k)$ that add to $m$ is $2^{O(m)}$ and since $k \le m$ for all tuples $B$, $g_m(\tilde{x}, z)$ can be computed in $n^{O(m)}$ time. For fixed $p, q, n,$ note that $g_m(\tilde{x}, z)$ is indeed only a function of $\tilde{x},$ $z$, and $m$, as the $w_{B, r}$'s are determined given $z$. Importantly, there is no dependence of $g$ on $x$. Then,
\[\BE[g_m(\tilde{x}, z)] = \left(\sum\limits_{i = 1}^{n} x_i z^i\right)^{m}.\]

    Finally, we provide uniform bounds on $f(\tilde{x}, w)$ that will give us our bounds on $g_m(\tilde{x}, z).$ If $|z| = 1$ and $|\arg z| \le \frac{2 \pi}{L}$, then by Proposition \ref{Complex3}, $\left|\frac{z-q}{p}\right| = 1 + O\left(\frac{1}{p^2 L^2}\right) = e^{O(1/(p^2 L^2))}.$ Therefore, since $b_r+\dots+b_k \le m,$ $|w_{B, r}| = e^{O(m^2/(p^2 L^2))},$ so $\left|w_{B, 1}^{i_1} \cdots w_{B, k}^{i_k-i_{k-1}}\right| = e^{O(m^2 n/(p^2 L^2))}$ whenever $1 \le i_1 < i_2 < \cdots < i_k \le n$. Since each $\tilde{x}_i$ is bounded by $1$ in absolute value, we have that
\[|f(\tilde{x}, w_B)| \le {n \choose k} \cdot e^{O(m^2 n/(p^2 L^2))} \le n^k \cdot e^{O(m^2 n/(p^2 L^2))},\]
    by the definition of $f$. Now, since $|z| = 1$ and $|w_{B, r}| \ge 1,$ we have that
\[|g_m(\tilde{x}, z)| \le \sum_{B} m! \cdot p^{-k} \cdot n^{k} \cdot e^{O(m^2 n/(p^2 L^2))} \le (p^{-1} mn)^{O(m)} \cdot e^{O(m^2 n/(p^2 L^2))},\]
    where we used the facts that ${m \choose b_1, \dots, b_k} \le m!$ and the number of tuples $B$ that we are summing over is at most $2^m$.

    If $|z - (1-\frac{p}{m})| \le \frac{p}{m},$ then by Proposition \ref{Complex2}, $|z^{b_r + \dots + b_k} - (1 - \frac{b_r + \dots + b_k}{m} \cdot p)| \le \frac{b_r + \dots + b_k}{m} \cdot p$ for all $|B| = k$ with $\sum b_i = m$ and all $1 \le r \le k.$ As $\frac{b_r + \dots + b_k}{m} \le 1,$ this implies that $|z^{b_r + \dots + b_k} - (1-p)| \le p,$ so $|w_{B, r}| \le 1$ for all $B$ and all $1 \le r \le k.$ Therefore, by the definition of $f$, $|f(\tilde{x}, (w_{B, 1}, \dots, w_{B, k}))| \le {n \choose k} \cdot |w_{B, 1} \cdots w_{B, k}|$ (since $i_1, i_2-i_1, \dots$ are all at least $1$ and $|\tilde{x}_i| \le 1$ for all $\tilde{x}$ and for all $i$). Moreover, note that $|z| \le 1,$ since $|z - (1-\frac{p}{m})| \le \frac{p}{m}.$ Thus,
\begin{align*}
|g_m(\tilde{x}, z)| &\le \sum\limits_{B} {m \choose b_1, \dots, b_k} \cdot p^{-k} \cdot \frac{1}{|w_{B, 1} \cdots w_{B, k}|} \cdot {n \choose k} \cdot |w_{B, 1} \cdots w_{B, k}| \\
&\le 2^m \cdot m! \cdot p^{-k} \cdot n^{m} \le (p^{-1} m n)^{O(m)},
\end{align*}
    as desired.
\end{proof}

Now, consider the problem of distinguishing between traces of strings $x^{(1)}, \dots, x^{(\ell)}$ that appear with probabilities $a_1, \dots, a_\ell,$ respectively, and traces of strings $y^{(1)}, \dots, y^{(\ell)}$ that appear with probabilities $b_1, \dots, b_\ell$, respectively. For a string $x = x_1 \cdots x_n,$ define $P(z; x) := \sum_{i = 1}^{n} x_i z^i$, which is a degree $n$ polynomial in $z$. We note that if $x$ came from the distribution over $x^{(1)}, \dots, x^{(\ell)}$,
\[\BE_{\tilde{x}}[g_k(\tilde{x}, z)] = \sum\limits_{i = 1}^{\ell} a_i P(z; x^{(i)})^k.\]
If $x$ came from the distribution over $y^{(1)}, \dots, y^{(\ell)}$, then
\[\BE_{\tilde{x}}[g_k(\tilde{x}, z)] = \sum\limits_{i = 1}^{\ell} b_i P(z; y^{(i)})^k.\]
We will prove that for some values of $k$ and $z$, $g_k(\tilde{x}, z)$ drawn from the former distribution and $g_k(\tilde{x}, z)$ drawn from the latter distribution have expected values that differ by a significant amount.

To do this, we first show the following result.

\begin{lemma} \label{Vandermonde}
    Let $a_1, \dots, a_\ell \in \BR$ be real numbers and $u_1, \dots, u_\ell \in \BC$ be distinct complex numbers. Also, fix some $\eps > 0$. Then, if $\sum |a_i| \ge \eps$, there must exist some $0 \le k < \ell$ such that 
\[\left|\sum_{i = 1}^{\ell} a_i u_i^k\right| \ge \frac{\eps}{\ell} \cdot \frac{\prod_{i > j} |u_i - u_j|}{\left(\sum_{i = 1}^{\ell} \sum_{k = 0}^{\ell-1} |u_i|^{2k}\right)^{(\ell-1)/2}}.\]
\end{lemma}

\begin{proof}
    Let $V \in \BC^{\ell \times \ell}$ be the matrix with $V_{i, k},$ the entry in the $i$th row and $k$th column of $V$, is $u_i^{k-1}.$ Then, note that
\[\sum_{i = 1}^{\ell} a_i u_i^k = \left(\left(\begin{matrix} a_1 & a_2 & \cdots & a_\ell \end{matrix}\right) \cdot V\right)_{k+1}\]
    for all $0 \le k \le \ell-1,$ where we note that $\left(\begin{matrix} a_1 & a_2 & \cdots & a_\ell \end{matrix}\right) \cdot V$ is a row vector. Therefore, it suffices to show that
\[\left\|\left(\begin{matrix} a_1 & a_2 & \cdots & a_\ell \end{matrix}\right) \cdot V\right\|_2 \ge \frac{\eps}{\sqrt{\ell}} \cdot \frac{\prod_{i > j} |u_i - u_j|}{\left(\sum_{i = 1}^{\ell} \sum_{k = 0}^{\ell-1} |u_i|^{2k}\right)^{(\ell-1)/2}}.\]
    Note that $\|a\|_2 \ge \frac{1}{\sqrt{\ell}} \sum |a_i| \ge \frac{\eps}{\sqrt{\ell}},$ so it suffices to show that the the smallest singular value of $V$ is at least 
\[\frac{\prod_{i > j} |u_i - u_j|}{\left(\sum_{i = 1}^{\ell} \sum_{k = 0}^{\ell-1} |u_i|^{2k}\right)^{(\ell-1)/2}}.\]
    However, note that the smallest singular value of $V$ is $\sqrt{\lambda_{\min}(V^\dagger V)},$ which is at least
\[\frac{\sqrt{\det (V^\dagger V)}}{\|V\|_F^{\ell-1}},\]
    since $\|V\|_F = \sqrt{tr(V^\dagger V)} \ge \sqrt{\lambda_{\max}(V^\dagger V)}.$ But $\sqrt{\det(V^\dagger V)} = |\det(V)| = \prod_{i > j} |u_i-u_j|$ (a well-known property of Vandermonde matrices) and $\|V\|_F^2 = \sum_{i = 1}^{\ell} \sum_{k = 0}^{\ell-1} |u_i|^{2k},$ so we are done.
\end{proof}

We note that Lemma \ref{Vandermonde} will be especially useful if we replace $u_i$ with $P(z; x^{(i)}).$ This motivates attempting to provide a lower bound for $\prod_{i > j} \left|P(z; x^{(i)}) - P(z; x^{(j)})\right|,$ which is precisely what we do in the next lemma.

\begin{lemma} \label{ComplexAnalysis}
    Fix $0 < p < 1$ and $m, L \ge 1.$ Let $x^{(1)}, \dots, x^{(\ell)}$ be distinct strings in $\{0, 1\}^n.$ Then, there exists some $z$ with norm $1$ and argument at most $\frac{2 \pi}{L}$ in magnitude such that 
\[\prod_{i > j} |P(z; x^{(i)}) - P(z; x^{(j)})| \ge n^{-\Theta(L \cdot \ell^2)}.\]
    Also, there exists some $z$ satisfying $|z - \left(1 - \frac{p}{m}\right)| \le \frac{p}{m},$ such that 
\[\prod_{i > j} |P(z; x^{(i)}) - P(z; x^{(j)})| \ge \exp\left(-\Theta\left(\ell^2 m^{1/3} p^{-1/3} n^{1/3} (\log n)^{2/3} + \ell^2 m p^{-1} \log n\right)\right).\]
\end{lemma}

\begin{proof}
    For each $i > j,$ let $Q_{ij}(z) = \frac{P(z; x^{(i)}) - P(z; x^{(j)})}{z^r},$ where we choose $r := r_{ij}$ as a nonnegative integer so that $Q_{ij}(z)$ is a polynomial but has nonzero constant coefficient. Note that $Q_{ij}$ has degree at most $n$ and has all coefficients in $\{-1, 0, 1\}$, with constant coefficient nonzero. Also, let $Q(z) = \prod_{i > j} Q_{ij}(z).$
    
    For the case $|z| = 1$ and $|\arg z| \le \frac{2 \pi}{L},$ it suffices to show that there is some $z$ in our range with $|Q(z)| \ge n^{-O(L \cdot \ell^2)}$, since $|z^r| = 1$. However, note that $Q_{ij}$ has degree at most $n$, all coefficients at most $1$ in absolute value, and constant coefficient $1$ in absolute value. Thus, $Q(z)$ has degree at most $\frac{\ell (\ell-1)}{2} \cdot n,$ all coefficients at most $n^{\ell(\ell-1)/2}$ in absolute value, and constant coefficient $1$ in absolute value. Thus, the lemma follows from Theorem \ref{Littlewood}.
    
    For the case of $|z - (1 - \frac{p}{m})| \le \frac{p}{m},$ choose some integer $L' \ge 3$ and consider the set of $z$ satisfying $|z| = 1 - \frac{p}{m \cdot L'}$ and $|\arg z| \le \frac{p}{m \cdot \sqrt{L'}}$. First, notice that that any such $z$ satisfies $\Re z \le 1 - \frac{p}{m \cdot L'}$ but $\Re z \ge 1 - \frac{p}{m \cdot L'} - \frac{p}{m \cdot \sqrt{L'}} \ge 1 - \frac{p}{m}.$ Moreover, note that $|\Im z| \le \frac{p}{m \cdot \sqrt{L'}}.$ Therefore,
\[\left|z - \left(1 - \frac{p}{m}\right)\right|^2 \le \left(\frac{p}{m} - \frac{p}{m \cdot L'}\right)^2 + \left(\frac{p}{m \cdot \sqrt{L'}}\right)^2 = \left(\frac{p}{m}\right)^2 \cdot \left[\left(1 - \frac{1}{L'}\right)^2 + \frac{1}{L'}\right] \le \left(\frac{p}{m}\right)^2.\]
    Thus, we only have to provide a lower bound in the range $|z| = 1 - \frac{p}{m \cdot L'}, |\arg z| \le \frac{p}{m \cdot \sqrt{L'}}$ for an appropriate $L'$. If we choose $z' := z/(1 - \frac{p}{m \cdot L'}),$ then $|z'| = 1$ and $|\arg z'| \le \frac{p}{m \cdot \sqrt{L'}}.$ Defining $Q'(z') := Q(z) = Q\left(z' \cdot (1 - \frac{p}{m \cdot L'})\right)$ means that $Q'$ has degree at most $\frac{\ell (\ell-1)}{2} \cdot n,$ all coefficients at most $n^{\ell(\ell-1)/2}$ in absolute value, and constant coeficient $1$ in absolute value. Thus, by Theorem \ref{Littlewood}, $|Q(z)| = |Q'(z')| \ge n^{-\Theta(\ell^2 \cdot m \cdot p^{-1} \cdot \sqrt{L'})}$ for some $z$ in the appropriate range. Finally, note that $|z| \ge e^{-\Theta(p/(m \cdot L'))},$ which means that $\prod |z^{r_{ij}}| \ge e^{-\Theta(p \cdot m^{-1} \cdot (L')^{-1} \cdot \ell^2 \cdot n)}.$ Thus, we have that
\[\prod\limits_{i > j} \left|P(z;x^{(i)})-P(z;x^{(j)})\right| \ge \exp\left(-\Theta(\ell^2 \cdot m \cdot p^{-1} \cdot \log n \cdot \sqrt{L'} + p \cdot m^{-1} \cdot \ell^2 \cdot n \cdot (L')^{-1})\right).\]
    Note that if $L' = \left((p^2 n)/(m^2 \log n)\right)^{2/3},$ then $\ell^2 \cdot m \cdot p^{-1} \cdot \log n \cdot \sqrt{L'}$ and $p \cdot m^{-1} \cdot \ell^2 \cdot n \cdot (L')^{-1}$ will be equal. However, it is possible that this means $L' \le 3$ in this case, in which case we choose $L' = 3.$ In the case where $L' = \left((p^2 n)/(m^2 \log n)\right)^{2/3} \ge 3,$ we have that 
\[\ell^2 \cdot m \cdot p^{-1} \cdot \log n \cdot \sqrt{L'} = p \cdot m^{-1} \cdot \ell^2 \cdot n \cdot (L')^{-1} = \ell^2 m^{1/3} p^{-1/3} n^{1/3} (\log n)^{2/3}.\]
    Otherwise, $L' = 3,$ and 
\[\ell^2 \cdot m \cdot p^{-1} \cdot \log n = \Theta\left(\ell^2 \cdot m \cdot p^{-1} \cdot \log n \cdot \sqrt{L'}\right) = \Omega\left(p \cdot m^{-1} \ell^2 n \cdot (L')^{-1}\right).\]
    This completes the case for $|z - \left(1 - \frac{p}{m}\right)| \le \frac{p}{m}.$
\end{proof}

We also note the following improvement to Lemma \ref{ComplexAnalysis}.

\begin{proposition} \label{Modification}
    Let $p, m, L, x^{(i)}$ be as in Lemma \ref{ComplexAnalysis}. Then, there exists some universal constant $C > 1$ and some $z \in \BC$ such that $|z| = 1, |\arg z| \le \frac{2 \pi}{L},$ $\arg z$ an integer multiple of $n^{-C \cdot L \cdot \ell^2}$, and
\[\prod_{i > j} |P(z; x^{(i)}) - P(z; x^{(j)})| \ge n^{-C \cdot L \cdot \ell^2}.\]
    Also, there exists $z \in \BC$ such that $|z - \left(1 - \frac{p}{m}\right)| \le \frac{p}{m},$ $\Re z, \Im z$ are both integer multiples of $\delta := \exp\left(-C\left(\ell^2 m^{1/3} p^{-1/3} n^{1/3} (\log n)^{2/3} + \ell^2 m p^{-1} \log n\right)\right)$, and
\[\prod_{i > j} |P(z; x^{(i)}) - P(z; x^{(j)})| \ge \delta = \exp\left(-C\left(\ell^2 m^{1/3} p^{-1/3} n^{1/3} (\log n)^{2/3} + \ell^2 m p^{-1} \log n\right)\right).\]
\end{proposition}

\begin{proof}
    Since $P(z; x^{(i)})-P(z; x^{(j)})$ has degree at most $n$ and all coefficients in $\{-1, 0, 1\},$ the polynomial $R(z) := \prod_{i > j} (P(z; x^{(i)})-P(z; x^{(j)}))$ has degree at most $n \cdot \frac{\ell(\ell-1)}{2}$ and all coefficients bounded by $n^{\ell^2}.$ Thus, for all $|z| \le 1,$ $\left|\frac{d}{dz} R(z)\right| \le n^{O(\ell^2)}.$
    
    Let $z$ be chosen as in Lemma \ref{ComplexAnalysis}. In the case of $|z| = 1, |\arg z| \le \frac{2 \pi}{L},$ we can replace $z$ with $z'$ such that $|z'| = 1$ and $\arg z'$ equals $\arg z$ rounded to the nearest multiple of $n^{-C \cdot L \cdot \ell^2}.$ (If $|\arg z'| \ge \frac{1}{L},$ we round in the opposite direction instead.) Then, $|R(z)-R(z')| \le n^{-C \cdot L \cdot \ell^2} \cdot n^{O(\ell^2)}$ and as $|R(z)| \ge n^{-\Theta(L \cdot \ell^2)}$ by Lemma \ref{ComplexAnalysis}, the first part of the proposition follows by the Triangle inequality for sufficiently large $C$.
    
    In the case of $|z - \left(1 - \frac{p}{m}\right)| \le \frac{p}{m},$ we can replace $z$ with $z'$ where $\Re z'$ is $\Re z$ rounded to the nearest multiple of $\delta$, and $\Im z'$ is $\Im z$ rounded to the nearest multiple of $\delta$. If $\left|z'-(1-\frac{p}{m})\right| > \frac{p}{m},$ then we can shift either $\Re z', \Im z'$ by a constant multiple of $\delta$ so that $\left|z'-(1-\frac{p}{m})\right| \le \frac{p}{m}.$ Then, $|R(z)-R(z')| \le \delta \cdot n^{O(\ell^2)}$, so if we choose $C$ to be sufficiently large, the second part of the proposition follows by the Triangle inequality for sufficiently large $C$.
\end{proof}

\begin{remark}
    We will use Lemma \ref{ComplexAnalysis} instead of Proposition \ref{Modification} in the next two results for simplicity. However, since we use Lemma \ref{ComplexAnalysis} as a black box, we can easily substitute in Proposition \ref{Modification}, which we do when finally proving Theorem \ref{main}.
\end{remark}

Now, we can prove the following result.

\begin{theorem} \label{Distinguish1}
    Fix $\eps > 0$ and $L \in \BN$, and let $\mathcal{D}$ and $\mathcal{D}'$ be $\ell$-sparse distributions over $\{0, 1\}^n$ with $d_{\text{TV}}(\mathcal{D}, \mathcal{D}') \ge \eps.$ Then, there exists some $1 \le k \le 2\ell-1$ and $z$ with $|z| = 1$ and $|\arg z| \le \frac{2 \pi}{L}$ such that
\[\left|\BE_{x \sim \mathcal{D}} \left[P(z; x)^k\right] - \BE_{x \sim \mathcal{D}'} \left[P(z; x)^k\right]\right| \ge \eps \cdot n^{-\Theta(L \cdot \ell^2)}.\]
    Moreover, for any $0 < p < 1$ and $m \ge 1,$ there exists some $1 \le k \le 2\ell-1$ and $z$ with $\left|z - \left(1 - \frac{p}{m}\right)\right| \le \frac{p}{m}$ such that
\[\left|\BE_{x \sim \mathcal{D}} \left[P(z; x)^k\right] - \BE_{x \sim \mathcal{D}'} \left[P(z; x)^k\right]\right| \ge \eps \cdot \exp\left(-\Theta\left(\ell^2 m^{1/3} p^{-1/3} n^{1/3} (\log n)^{2/3} + \ell^2 m p^{-1} \log n \right)\right).\]
\end{theorem}

\begin{proof}
    Suppose that $x \sim \mathcal{D}$ equals $x^{(j)}$ with probability $a_j$ for $1 \le j \le \ell$ and $y \sim \mathcal{D}$ equals $y^{(j)}$ with probability $b_j$ for $1 \le j \le \ell.$ Then, we can write
\[\BE_{x \sim \mathcal{D}} \left[P(z; x)^k\right] - \BE_{x \sim \mathcal{D}'} \left[P(z; x)^k\right] = \sum\limits_{j = 1}^{\ell} a_j P(z; x^{(j)})^k - \sum\limits_{j = 1}^{\ell} b_j P(z; y^{(j)})^k.\]
    Since $d_{\text{TV}}(\mathcal{D}, \mathcal{D}') \ge \eps$, we can rewrite this as
\[\sum\limits_{j = 1}^{\ell'} c_j P(z; x^{(j)})^k\]
    for some $\ell' \le 2 \ell,$ the $x^{(j)}$'s distinct strings coming from the original $x^{(j)}$ and $y^{(j)}$ strings, and $\sum c_i = 0,$ $\sum |c_i| \ge \eps$. By Lemma \ref{Vandermonde}, for any $z$, there exists some $0 \le k < \ell'$ such that
\[\left|\sum\limits_{j = 1}^{\ell'} c_i P(z; x^{(j)})^k\right| \ge \frac{\eps}{\ell'} \cdot \frac{\prod_{i > j} |P(z; x^{(i)}) - P(z; x^{(j)})|}{\left(\sum_{i = 1}^{\ell'} \sum_{k = 0}^{\ell'-1} |P(z; x^{(i)})|^{2k}\right)^{(\ell'-1)/2}}.\]
    Note that $k$ cannot equal $0$, however, since $\sum c_i = 0.$ Thus, $1 \le k \le 2\ell-1.$ 
    
    For the first half of the result, choose $z$ based on Lemma \ref{ComplexAnalysis}, so that $|z| = 1,$ $|\arg z| \le \frac{2 \pi}{L},$ and $\prod_{i > j} |P(z; x^{(i)}) - P(z; x^{(j)})| \ge n^{-\Theta(L \cdot (\ell')^2)},$ and then choose $k$ accordingly. Since $|P(z; x)| \le n$ for all $|z| = 1$ and $x \in \{0, 1\}^n,$ we have that
\begin{align*}\frac{\eps}{\ell'} \cdot \frac{\prod_{i > j} |P(z; x^{(i)}) - P(z; x^{(j)})|}{\left(\sum_{i = 1}^{\ell'} \sum_{k = 0}^{\ell'-1} |P(z; x^{(i)})|^{2k}\right)^{(\ell'-1)/2}} &\ge \frac{\eps}{\ell'} \cdot \frac{n^{-\Theta(L \cdot (\ell')^2)}}{\left((\ell')^2 \cdot n^{2 \ell'}\right)^{(\ell'-1)/2}} \\
&\ge \eps \cdot n^{-\Theta(L \cdot \ell^2)} \cdot \ell^{-\Theta(\ell)} \\
&= \eps \cdot n^{-\Theta(L \cdot \ell^2)}
\end{align*}
    since $\ell' \le 2 \ell.$
    
    For the second half, we choose $z$ based on Lemma \ref{ComplexAnalysis} such that $\left|z - \left(1 - \frac{p}{m}\right)\right| \le \frac{p}{m},$ so that 
\[\prod_{i > j} |P(z; x^{(i)}) - P(z; x^{(j)})| \ge \exp\left(-\Theta\left((\ell')^2 m^{1/3} p^{-1/3} n^{1/3} (\log n)^{2/3} + (\ell')^2 m p^{-1} \log n\right)\right),\]
    and then we can choose $k$ accordingly. Noting that $\ell' \le \exp\left(\Theta\left(\ell^2 m^{1/3} p^{-1/3} n^{1/3} (\log n)^{2/3}\right)\right)$ and 
\begin{align*}
    \left(\sum_{i = 1}^{\ell'} \sum_{k = 0}^{\ell'-1} |P(z; x^{(i)})|^{2k}\right)^{(\ell'-1)/2} &\le \left((\ell')^2 \cdot n^{2 \ell'}\right)^{(\ell'-1)/2} \\
    &\le \exp\left(\Theta(\log n \cdot \ell^2)\right) \\
    &\le \exp\left(\Theta\left(\ell^2 m^{1/3} p^{-1/3} n^{1/3} (\log n)^{2/3}\right)\right),
\end{align*}
    as $\ell' \le 2\ell,$ the result is immediate.
\end{proof}

Previously, we had only been dealing with distinguishing between two fixed distributions, though in the population recovery problem, we need to recover the original distribution, which means there can be a very large number of distributions to choose from. We will use Lemma \ref{CreatingG} and Theorem \ref{Distinguish1} to prove the final result, but we will also use the modification of Proposition \ref{Modification} so that the set of $z$ we deal with is small enough that we can get good estimates of $\BE[g_k(\tilde{x}, z)]$ for all $z, k$.

\begin{lemma} \label{Algorithm}
    Fix $0 < \eps < 1$ and reals $M \ge \eps^{-1}, N \ge 1.$ Let $S$ be a finite set of complex numbers, such that $|g_k(\tilde{x}, z)| \le N$ for all $\tilde{x}$ and all $z \in S,$ but for any pair of $\ell$-sparse distributions $\mathcal{D}$ and $\mathcal{D}'$ over $\{0, 1\}^n$ with $d_{\text{TV}}(\mathcal{D}, \mathcal{D}') \ge \eps$, there exists $z \in S$ and $0 \le k \le 2\ell-1$ such that $|\BE_{\tilde{x}: x \sim \mathcal{D}}[g_k(\tilde{x}, z)] - \BE_{\tilde{x}: x \sim \mathcal{D}'}[g_k(\tilde{x}, z)]| \ge M^{-1}$. Then, the population recovery problem can be solved in $O(M^2 \cdot N^2 \cdot \log (\ell \cdot |S| \cdot \delta^{-1}))$ queries with at least $1-\delta$ probability.
\end{lemma}

\begin{proof}
    For any fixed $s \in S$ if $x$ were drawn from from some distribution $\mathcal{D},$ the empirical mean of $\Theta\left(M^2 \cdot N^2 \cdot \log (\ell |S| \delta^{-1})\right)$ samples of $g_k(\tilde{x}, z)$ would be at most $\frac{1}{4M}$ away from $\BE_{\tilde{x}: x \sim \mathcal{D}}[g_k(\tilde{x}, z)]$ in magnitude with probability at least $1 - \frac{\delta}{2 \ell \cdot |S|}$ by the Chernoff bound. Thus, by the union bound, with probability at least $1-\delta,$ we will have an estimate of $\BE_{\tilde{x}:x \sim \mathcal{D}} P(z; x)^k$ within an additive error of $\frac{1}{4M}$ for all $z \in S$ and all $0 \le k \le 2 \ell-1$.
    
    Now, to recover $\mathcal{D},$ we simply output any distribution $\mathcal{D}_0$ such that $\BE_{\tilde{x}: x \sim \mathcal{D}_0}[g_k(\tilde{x}, z)] = \BE_{x \sim \mathcal{D}_0} P(z; x)^k$ is within a $\frac{1}{4M}$ additive factor from our empirical mean of $g_k(\tilde{x}, z)$ for all $z \in S$ and all $0 \le k \le 2 \ell-1$. Note that $\mathcal{D}_0 = \mathcal{D}$ works, so there exists a solution. However, if we output some $\mathcal{D}_0,$ note that $d_{TV}(\mathcal{D}, \mathcal{D}_0) \le \eps,$ or else there is some $z \in S,$ $0 \le k \le 2\ell-1$ such that $\left|\BE_{x \sim \mathcal{D}_0} P(z; x)^k - \BE_{x \sim \mathcal{D}} P(z; x)^k\right| \ge \frac{1}{M},$ which would mean that $\BE_{x \sim \mathcal{D}_0} P(z; x)^k$ differs from the empirical mean of $g_k(\tilde{x}, z)$ by at least $\frac{3}{4M}$ for some $z \in S, 0, \le k \le 2\ell-1.$ Thus, with probability at least $1-\delta$, we output $\mathcal{D}_0$ such that $d_{TV}(\mathcal{D}, \mathcal{D}_0) \le \eps.$
\end{proof}

We can now finish the proof of Theorem \ref{main}.

\begin{proof}[Proof of Theorem \ref{main}]
    Let $\mathcal{D}, \mathcal{D}'$ be distributions such that $d_{TV}(\mathcal{D}, \mathcal{D}') \ge \eps.$ By Theorem \ref{Distinguish1} and the modification of Proposition \ref{Modification}, for any integer $L \ge 1,$ there exists a universal constant $C$ and $0 \le k \le 2 \ell-1, z \in \BC$ such that $|z| = 1$, $|\arg z| \le \frac{2 \pi}{L}$, $\arg z$ is an integer multiple of $n^{-C \cdot L \cdot \ell^2},$ and
\[|\BE_{\tilde{x}: x \sim \mathcal{D}}[g_k(\tilde{x}, z)] - \BE_{\tilde{x}: x \sim \mathcal{D}'}[g_k(\tilde{x}, z)]| = \left|\BE_{x \sim \mathcal{D}} \left[P(z; x)^k\right] - \BE_{x \sim \mathcal{D}'} \left[P(z; x)^k\right]\right| \ge \eps \cdot n^{-C \cdot L \cdot \ell^2} := M^{-1}.\]

    Now, our set $S$ will be all $z$ such that $|z| = 1, |\arg z| \le \frac{2 \pi}{L}$, and $\arg z$ is an integer multiple of $n^{-C \cdot L \cdot \ell^2}$. Note that $|S| = n^{O(L \cdot \ell^2)}$, and for all $\tilde{x}$, $|g_k(\tilde{x}, z)| \le (p^{-1} \ell n)^{O(\ell)} \cdot e^{O(\ell^2 n/(p^2 L^2))} := N$ by Lemma \ref{CreatingG}. We will choose $L = \left\lfloor\left(\frac{n}{\log n \cdot p^2}\right)^{1/3} \right\rfloor,$ so that $L \cdot \ell^2 \cdot \log n = \Theta\left(\frac{\ell^2 n}{p^2 L^2}\right) = \Theta\left(n^{1/3} \cdot (\log n)^{2/3} \cdot \ell^{2} \cdot p^{-2/3}\right).$ Note that since $p < 1,$ we have that $L \ge 1.$ Therefore, the number of queries, by Lemma \ref{Algorithm}, is at most
\[\eps^{-2} \cdot \left(p^{-1} \ell n\right)^{\Theta(\ell)} \cdot e^{\Theta(n^{1/3} (\log n)^{2/3} \ell^{2} p^{-2/3})} = \eps^{-2} \cdot \exp\left(\Theta(n^{1/3} (\log n)^{2/3} \ell^{2} p^{-2/3})\right).\]
    This completes the proof of Equation \eqref{MainEq1} of Theorem \ref{main}.
    
    Next, by Theorem \ref{Distinguish1} and the modification of Proposition \ref{Modification}, there exist $0 \le k \le 2\ell-1$ and $z \in \BC$ with $\left|z - (1 - \frac{p}{2\ell})\right| \le \frac{p}{2 \ell}$, $\Re z, \Im z$ integer multiples of $\exp\left(-C\left(\ell^{7/3} p^{-1/3} n^{1/3} (\log n)^{2/3} + \ell^3 p^{-1} \log n \right)\right)$, and
\[\left|\BE_{\tilde{x}: x \sim \mathcal{D}} \left[g_k(\tilde{x}, z)\right] - \BE_{\tilde{x}: x \sim \mathcal{D}'} \left[g_k(\tilde{x}, z)\right]\right| \ge \eps \cdot \exp\left(-C\left(\ell^{7/3} p^{-1/3} n^{1/3} (\log n)^{2/3} + \ell^3 p^{-1} \log n \right)\right) := M^{-1}.\]
    Note that the corresponding set $S$ has size $\exp\left(C\left(\ell^{7/3} p^{-1/3} n^{1/3} (\log n)^{2/3} + \ell^3 p^{-1} \log n \right)\right)$. Now, by Lemma \ref{CreatingG}, for all $0 \le k \le 2\ell-1$ and all $\tilde{x},$ $|g_k(\tilde{x}, z)| \le (p^{-1} \ell n)^{\Theta(\ell)} := N.$ Thus, the number of queries, by Lemma \ref{Algorithm}, is at most
\[\eps^{-2} \cdot \exp\left(\Theta\left(\ell^{7/3} p^{-1/3} n^{1/3} (\log n)^{2/3} + \ell^3 p^{-1} \log n \right)\right) \cdot \exp\left(\Theta\left(\ell (\log p^{-1} + \log \ell + \log n)\right)\right)\]
\[= \eps^{-2} \cdot \exp\left(\Theta\left(\ell^{7/3} p^{-1/3} n^{1/3} (\log n)^{2/3} + \ell^3 p^{-1} \log n \right)\right).\]
    This completes the proof of Equation \eqref{MainEq2} of Theorem \ref{main}.
\end{proof}

\section{Faster Algorithm} \label{AlgSection}

In this section, we prove Theorem \ref{MainAlg}. Recall that our goal is to determine an unknown $\ell$-sparse distribution $\mathcal{D}$ that equals some $x^{(i)} \in \{0, 1\}^n$ with probability $a_i$ for all $1 \le i \le \ell.$ Like in Section \ref{PopulationRecoveryBounds}, we will choose a complex number $z$ and try to determine $P(z; x^{(i)}),$ where we recall that for $x \in \{0, 1\}^n,$ $P(z; x) = \sum_{i = 1}^{n} x_i z^i.$

To do this, we will first find good estimates for the symmetric polynomials of $P(z; x^{(i)}).$ First, for a distribution $\mathcal{D}$ with support size exactly $\ell$, and for any $1 \le k \le \ell,$ define 
\[\sigma_k(z; \mathcal{D}) := \sum\limits_{1 \le i_1 < \cdots < i_k \le \ell} P(z; x^{(i_1)}) \cdots P(z; x^{(i_k)}).\]
In other words, $\sigma_k(z; \mathcal{D})$ is the $k$th elementary symmetric polynomial of $P(z; x^{(1)}), \dots, P(z; x^{(\ell)}).$ Also, we will use $P_i$ to mean the polynomial $P(z; x^{(i)})$, and we define $\sigma_k(\mathcal{D})$ to be the polynomial $Q$ such that $Q(z) = \sigma_k(z; \mathcal{D})$, so $\sigma_k(\mathcal{D})$ is the $k$th symmetric polynomial of $P_1, \dots, P_\ell$.

\begin{remark}
    Note that $\sigma_k(z; \mathcal{D})$ and $\sigma_k(\mathcal{D})$ do not depend on the mixture weights of $\mathcal{D},$ but just depend on the elements $x^{(1)}, \dots, x^{(\ell)}$ in the support of $\mathcal{D}.$ 
\end{remark}

In Subsection \ref{SymmetricPoly1}, we will provide good estimates for $\sigma_k(z; \mathcal{D})$ for certain values of $z.$ In Subsection \ref{SymmetricPoly2}, we will use these estimates to actually determine the polynomials $\sigma_k(\mathcal{D}).$ Finally, in Subsection \ref{FinishAlgorithm}, we determine the strings $x^{(i)}$, and then we determine the probabilities $a_i$ of each $x^{(i)}$ being drawn from $\mathcal{D}$, and also deal with the case where $\mathcal{D}$ may have support size less than $\ell$.

\subsection{Estimates for the Symmetric Polynomials} \label{SymmetricPoly1}

In Section \ref{PopulationRecoveryBounds}, we were able to get estimates for $\sum a_i P(z; x^{(i)})^k$ for all $0 \le k \le 2\ell-1.$ First, assuming we know these values exactly, we show how to determine $\sigma_k(z; \mathcal{D})$ for all $1 \le k \le \ell.$ For simplicity, we assume the distribution is exactly $\ell$-sparse, i.e. $\mathcal{D}$ has support size exactly $\ell$, in this subsection.

The following proposition shows how to reconstruct elementary symmetric polynomials of some variables $u_1, \dots, u_{\ell}$ given weighted sums of $u_i^k$ for all $1 \le k \le 2 \ell - 1.$ We note this method is the same as Prony's method, commonly used in signal processing and other areas, though we include the proof here for completeness. The main challenge, however, is showing this method is sufficiently \emph{robust} to error in our setting.

\begin{proposition} \label{Recursion}
    Let $a_1, \dots, a_\ell \in \BR^+$ be positive reals and let $u_1, \dots, u_\ell \in \BC$ be distinct complex numbers. Let $b_k := \sum_{i = 1}^{\ell} a_i u_i^k$ for $0 \le k \le 2 \ell -1.$ Furthermore, let $r_k := \sigma_k(u_1, \dots, u_\ell)$ be the $k$th symmetric polynomial of $u_1, \dots, u_\ell,$ multiplied by $(-1)^{k-1},$ i.e.
\[r_k = (-1)^{k-1} \sum\limits_{1 \le i_1 < \dots < i_k \le \ell} u_{i_1} \cdots u_{i_k}.\]
    Then, for all $0 \le k \le \ell-1$,
\[b_{k + \ell} = \sum\limits_{j = 1}^{\ell} r_j b_{k + \ell-j}.\]
\end{proposition}

\begin{proof}
    Note that for all $i$, $u_i$ is a root of $\prod_{j = 1}^{\ell} (x-u_j) = 0$, but $\prod (x-u_j)$ can be written as $x^\ell - r_1 x^{\ell-1} - r_2 x^{\ell-2} - \cdots - r_\ell.$ Therefore, we have that
\[u_i^\ell = \sum\limits_{j = 1}^{\ell} r_j u_i^{\ell-j}.\]
    Thus,
\[b_{k+\ell} = \sum\limits_{i = 1}^{\ell} a_i u_i^{k} u_i^\ell = \sum\limits_{i = 1}^{\ell} a_i u_i^{k} \left(\sum\limits_{j = 1}^{\ell} r_j u_i^{\ell-j}\right)\]
    and by swapping the sums, we get that this equals
\[\sum\limits_{j = 1}^{\ell} r_j \left(\sum\limits_{i = 1}^{\ell} a_i u_i^{k+\ell-j}\right) = \sum\limits_{j = 1}^{\ell} r_j b_{k+\ell-j}. \qedhere\]
\end{proof}

Note that the above proposition can be rewritten as the following matrix identity:

\[\left(\begin{matrix} b_0 & b_1 & \cdots & b_{\ell-1} \\ b_1 & b_2 & \cdots & b_{\ell} \\ \vdots & \vdots & \ddots & \vdots \\ b_{\ell-1} & b_\ell & \cdots & b_{2\ell-2} \end{matrix}\right) \cdot \left(\begin{matrix} r_\ell \\ r_{\ell-1} \\ \vdots \\ r_1 \end{matrix}\right) = \left(\begin{matrix} b_\ell \\ b_{\ell+1} \\ \vdots \\ b_{2\ell-1} \end{matrix}\right)\]

Thus, by treating $u_i$ as $P(z; x^{(i)}),$ we can determine the value of $r_i$ and therefore the value of $\sigma_i(z; \mathcal{D}).$ However, we are not actually given the $b_i$'s exactly, but with sufficiently many samples can get good estimates $\tilde{b}_i$ of each $b_i$. We will essentially show that determining the $r_i$'s from the $b_i$'s is \emph{robust}, meaning that we will get good estimates for the $r_i$'s assuming we have sufficiently good estimates for the $b_i$'s. First, we will need the following simple result.

\begin{proposition} \label{EasyMatrix}
    Let $\{a_i\}, \{u_i\}, \{b_k\}$ be as in Proposition \ref{Recursion}. Let $B$ be the matrix with $B_{ij} = b_{i+j-2},$ as above, and let $v \in \BC^\ell$ be the vector with $v_i = b_{\ell-1+i}.$ Also, let $w \in \BC^\ell$ be the vector with $w_i = r_{\ell+1-i}$. As noted above, $B \cdot w = v$.
    
    Then, $B = V^T \cdot A \cdot V$, where $V \in \BC^{\ell \times \ell}$ is the Vandermonde matrix $V_{i, k} = u_i^{k-1},$ and $A = \text{diag}(a_i)$ is the diagonal matrix in $\BR^{\ell \times \ell}$ with $a_i$ in the $i$th row, $i$th column. Also, $v = V^T \cdot U^\ell \cdot a,$ where $U = \text{diag}(u_i)$ is the diagonal matrix in $\BC^{\ell \times \ell}$ with $u_i$ in the $i$th row, $i$th column, and $a$ is the column vector with $i$th entry $a_i$.
\end{proposition}

\begin{remark}
    Note that $V^T$ is the transpose and not necessarily the conjugate transpose of $V$. 
\end{remark}

\begin{proof}
    The proof can essentially be seen by writing out the matrices. By definition of $b_k,$ we have
\[\left(\begin{matrix} b_0 & b_1 & \cdots & b_{\ell-1} \\ b_1 & b_2 & \cdots & b_{\ell} \\ \vdots & \vdots & \ddots & \vdots \\ b_{\ell-1} & b_\ell & \cdots & b_{2\ell-2} \end{matrix}\right) = \left(\begin{matrix} 1 & 1 & \cdots & 1 \\ u_1 & u_2 & \cdots & u_\ell \\ \vdots & \vdots & \ddots & \vdots \\ u_1^{\ell-1} & u_2^{\ell-1} & \cdots & u_\ell^{\ell-1} \end{matrix}\right) \cdot \left(\begin{matrix} a_1 & 0 & \cdots & 0 \\ 0 & a_2 & \cdots & 0 \\ \vdots & \vdots & \ddots & \vdots \\ 0 & 0 & \cdots & a_\ell \end{matrix}\right) \cdot \left(\begin{matrix} 1 & u_1 & \cdots & u_1^{\ell-1} \\ 1 & u_2 & \cdots & u_2^{\ell-1} \\ \vdots & \vdots & \ddots & \vdots \\ 1 & u_\ell & \cdots & u_\ell^{\ell-1} \end{matrix}\right),\]
    but the left hand side is $B$ and the right hand side is $V^T A V.$ Likewise, we have
\[\left(\begin{matrix} b_\ell \\ b_{\ell+1} \\ \vdots \\ b_{2\ell-1}\end{matrix}\right) = \left(\begin{matrix} 1 & 1 & \cdots & 1 \\ u_1 & u_2 & \cdots & u_\ell \\ \vdots & \vdots & \ddots & \vdots \\ u_1^{\ell-1} & u_2^{\ell-1} & \cdots & u_\ell^{\ell-1} \end{matrix}\right) \cdot \left(\begin{matrix} u_1^\ell & 0 & \cdots & 0 \\ 0 & u_2^\ell & \cdots & 0 \\ \vdots & \vdots & \ddots & \vdots \\ 0 & 0 & \cdots & u_\ell^\ell \end{matrix}\right) \cdot \left(\begin{matrix} a_1 \\ a_2 \\ \vdots \\ a_\ell \end{matrix}\right),\]
    but the left hand side is $v$ and the right hand side is $V^T \cdot U^\ell \cdot a.$
\end{proof}

We will only be able to obtain a good estimate of $r_1, \dots, r_\ell$ given our estimates $\tilde{b}_0, \dots, \tilde{b}_{2 \ell-1}$ assuming that the Vandermonde matrix $V$ is sufficiently ``nice'', i.e. we will need all singular values of both $V$ and $V^T A V$ to not be too small. One part of our algorithm, explained in the following lemma, will thus be to determine if $V$ is nice when we are only given the estimates $\tilde{b}_0, \dots, \tilde{b}_{2\ell-1}.$

\begin{lemma} \label{Alg1}
    Suppose $V \in \BC^{\ell \times \ell}$ is some unknown matrix, $A \in \BR^{\ell \times \ell}$ is an unknown diagonal matrix with all diagonal entries positive reals, $E \in \BC^{\ell \times \ell}$ is some unknown ``error'' matrix, and we are given $\tilde{B} = V^T A V + E$. Let $\alpha$ be known such that $\alpha \le \min A_{ii} \le 2 \alpha,$ and let $\beta$ be known such that $\beta \le \det A \le 2 \beta.$ Let $\delta < 1$ be some known parameter, and suppose we also know that $\|E\| \le \frac{\alpha \cdot \delta}{4 \ell}$. Finally, let $\sigma_{min}(V^T A V)$ be the least singular value of $V^T A V.$ Then, there exists a polynomial (in $\ell$) time algorithm that will always return YES if $|\det V| \ge \delta$ and $\sigma_{min}(V^T A V) \ge \alpha \cdot \delta,$ but will always return NO if either $|\det V| < \frac{\delta}{3}$ or $\sigma_{min}(V^T A V) < \frac{\alpha \cdot \delta}{2}.$
\end{lemma}

\begin{proof}
    First, we compute the least singular value of $\tilde{B} = V^T A V + E$ and check if it is at least $\frac{3}{4} \alpha \delta.$ If the smallest singular value of $V^T A V + E$ is less than $\frac{3}{4} \alpha \delta,$ then by Theorem \ref{Weyl}, the smallest singular value of $V^T A V$ is less than $\alpha \delta,$ so our algorithm will return NO. Otherwise, we know that the smallest singular value of $V^T A V$ is at least $\frac{\alpha \cdot \delta}{2}.$
    
    Now, let $\sigma_1 \ge \cdots \ge \sigma_\ell \ge \frac{\alpha \cdot \delta}{2}$ be the singular values of $V^T A V$ and let $\tau_1 \ge \cdots \ge \tau_\ell$ be the singular values of $V^T A V + E.$ Then, by Theorem \ref{Weyl}, $|\sigma_i-\tau_i| \le \frac{\alpha \cdot \delta}{4 \ell} \le \frac{\sigma_i}{2 \ell},$ so $1 - \frac{1}{2 \ell} \le \frac{\sigma_i}{\tau_i} \le 1 + \frac{1}{2 \ell}.$ Therefore,
\[\frac{1}{2} \le \left(1 - \frac{1}{2 \ell}\right)^\ell \le \frac{\prod \sigma_i}{\prod \tau_i} = \frac{|\det (V^T A V)|}{|\det (V^T A V + E)|} \le \left(1 + \frac{1}{2 \ell}\right)^\ell \le 2.\]
    Now, if our algorithm hasn't returned NO already, we compute the determinant of $\tilde{B} = V^T A V + E$ and check if its magnitude is at least $\frac{\beta \cdot \delta^2}{2}.$ If not, then we know $|\det(V^T A V)| < \beta \cdot \delta^2,$ but since $\det(A) \ge \beta,$ this means that $|\det(V)| < \delta$, so we return NO. Otherwise, $|\det (V^T A V)| \ge \frac{\beta \cdot \delta^2}{4},$ but since $\det A \le 2 \beta,$ we have that $|\det V| \ge \frac{\delta}{\sqrt{8}} \ge \frac{\delta}{3}.$ Since we already know that $\sigma_{min}(V^T A V) \ge \frac{\alpha \cdot \delta}{2},$ we can return YES.
\end{proof}

Now, assuming $\sigma_{min}(V)$ and $\sigma_{min}(V^T A V)$ are not too small, we show how to get a good estimate for $w = B^{-1} v.$

\begin{lemma} \label{UsingWoodbury}
    Let $a_1, \dots, a_\ell \in \BR^+$ be positive reals that add up to $1$, and $\{u_i\}, \{b_k\}, \{r_k\}$ be defined as in Proposition \ref{Recursion}. Let the matrices $A, B, U, V$ and vectors $a, v, w$ be defined as in Proposition \ref{EasyMatrix}.
    
    Then, $B$ is invertible. Now, let $\alpha$ be a known parameter such that $\alpha \le \min a_i \le 2 \alpha.$ Also, set $0 < \gamma, \eta < 1$ to be known parameters such that $\sigma_{min}(V) \ge \gamma \cdot (\max |u_i|)^\ell$ and $\sigma_{min}(V^T A V) \ge \alpha \cdot \gamma$. 
    Suppose that $\tilde{b}_0, \dots, \tilde{b}_{2 \ell-1}$ are estimates such that $|\tilde{b}_i-b_i| \le \frac{\alpha \cdot \gamma^2 \cdot \eta}{4\ell^2}$ for all $i$, and define $\tilde{B}, \tilde{v}$ analogously to $B, v.$ Then, $\tilde{B}$ is invertible, and if we define $\tilde{w} := \tilde{B}^{-1} \tilde{v}$, then $\|\tilde{w}-w\|_2 \le \eta$.
\end{lemma}

\begin{proof}
    First, since the $u_i$'s are all distinct and $a_i$'s are all positive, this means $V, A$ are both invertible, so $V^T A V = B$ is invertible also.
    
    Now, let $E = \tilde{B}-B$, and let $e = \tilde{v}-v$. Then, $\tilde{B} = V^T A V + E$ and $\tilde{v} = V^T U^\ell a + e$ by Proposition \ref{EasyMatrix}. Since all entries in $E$ and all entries in $e$ are at most $\frac{\alpha \cdot \gamma^2 \cdot \eta}{4 \ell^2}$ in magnitude, we have that $\|E\| \le \|E\|_F \le \frac{\alpha \cdot \gamma^2 \cdot \eta}{4 \ell},$ and $\|e\|_2 \le \frac{\alpha \cdot \gamma^2 \cdot \eta}{4 \ell}$ also. Therefore, by Theorem \ref{Weyl}, $\sigma_{min}(V^T A V + E) \ge \alpha \cdot \gamma - \frac{\alpha \cdot \gamma^2 \cdot \eta}{4 \ell} \ge \frac{\alpha \cdot \gamma}{2},$ and thus $\tilde{B}$ is invertible.
    
    Next, note that
\[\|\tilde{w}-w\|_2 = \|\tilde{B}^{-1} \tilde{v} - B^{-1} v\|_2 \le \|\tilde{B}^{-1} (\tilde{v}-v)\|_2 + \|(\tilde{B}^{-1} - B^{-1}) v\|_2 \le \|\tilde{B}^{-1}\| \cdot \|e\|_2 + \|(\tilde{B}^{-1} - B^{-1}) v\|_2.\]
    We can bound $\|\tilde{B}^{-1}\| \cdot \|e\|_2$ since $\|\tilde{B}^{-1}\| = \sigma_{min}(\tilde{B})^{-1} \le \frac{2}{\alpha \cdot \gamma}$ and $\|e\|_2 \le \frac{\alpha \cdot \gamma^2 \cdot \eta}{4 \ell},$ so $\|\tilde{B}^{-1}\| \cdot \|e\|_2 \le \frac{\gamma \cdot \eta}{2 \ell} \le \frac{\eta}{2}.$ To bound $\|(\tilde{B}^{-1} - B^{-1}) v\|_2$, we note that  by Theorem \ref{Woodbury},
\begin{align*}
    (\tilde{B}^{-1}-B^{-1})v &= \left((V^T A V + E)^{-1} - (V^T A V)^{-1}\right) (V^T U^\ell a) \\
    &= -(V^T A V)^{-1} E (I + (V^T A V)^{-1} E)^{-1} (V^T A V)^{-1} V^T U^\ell a.
\end{align*}
    Note that the right part of the last line, $(V^T A V)^{-1} V^T U^\ell a,$ can be expanded as $V^{-1} A^{-1} (V^T)^{-1} V^T U^\ell a = V^{-1} A^{-1} U^\ell a.$ Since $A^{-1}, U$ are diagonal, they commute, so this equals $V^{-1} U^\ell A^{-1} a = V^{-1} U^\ell \textbf{1},$ where $\textbf{1}$ is the $\ell$-dimensional vector of all $1$s. Thus,
\[\|(\tilde{B}^{-1}-B^{-1})v\|_2 \le \|(V^T A V)^{-1}\| \cdot \|E\| \cdot \|(I + (V^T A V)^{-1} E)^{-1}\| \cdot \|V^{-1}\| \cdot (\max |u_i|)^\ell \cdot \sqrt{\ell},\]
    where we used $\|U\| = \max |u_i|$ since $U$ is diagonal. Now, since $\|(V^T A V)^{-1}\| = \sigma_{min}(V^T A V)^{-1} \le \frac{1}{\alpha \cdot \gamma}$ and $\|E\| \le \frac{\alpha \cdot \gamma^2 \cdot \eta}{4 \ell},$ we have that all singular values of $I + (V^T A V)^{-1} E$ are between $1 - \frac{\gamma \cdot \eta}{4 \ell}$ and $1 + \frac{\gamma \cdot \eta}{4 \ell}.$ Thus, $\|(I + (V^T A V)^{-1} E)^{-1}\| \le 2.$ Therefore, 
\[\|(\tilde{B}^{-1}-B^{-1})v\|_2 \le \frac{1}{\alpha \cdot \gamma} \cdot \frac{\alpha \cdot \gamma^2 \cdot \eta}{4 \ell} \cdot 2 \cdot \frac{1}{\gamma \cdot (\max |u_i|)^\ell} \cdot (\max |u_i|)^\ell \cdot \sqrt{\ell} = \frac{\eta}{2 \sqrt{\ell}} \le \frac{\eta}{2}.\]
    Adding the errors gives us $\|\tilde{w}-w\|_2 \le \frac{\eta}{2} + \frac{\eta}{2} \le \eta.$
\end{proof}

Using Lemma \ref{Alg1} and Lemma \ref{UsingWoodbury}, we show how to approximate $r_1, \dots, r_\ell$ using approximations for $b_0, \dots, b_{2\ell-2}.$

\begin{lemma} \label{Alg2}
    Let $L = \left\lfloor\left(\frac{n}{\log n \cdot p^2}\right)^{1/3} \right\rfloor$, and let $z \in \BC$ be a known complex number with $|z| = 1$ and $|\arg z| \le \frac{1}{L}$. Suppose we are given sample access to traces from an unknown $\ell$-sparse distribution $\mathcal{D}$ that equals $x^{(i)} \in \{0, 1\}^n$ with probability $a_i$ for distinct $x^{(1)}, \dots, x^{(\ell)},$ where $a_i$'s are positive reals that add to $1$. Now, let $u_i := P(z; x^{(i)}),$ and let $\{b_k\}, \{r_k\}, A, B, U, V, a, v, w$ be as in Proposition \ref{UsingWoodbury}. Finally, suppose we are also given constants $\alpha, \beta < 1$ such that $\alpha \le \min a_i \le 2 \alpha$ and $\beta \le \prod a_i \le 2 \beta.$
    
    Then, for any fixed constant $C$, there exists an algorithm taking $\alpha^{-2} \cdot \exp\left(O\left(n^{1/3} (\log n)^{2/3} \ell^2 p^{-2/3}\right)\right)$ time and queries that either outputs nothing or, for all $1 \le j \le \ell$, outputs an estimate for $\sigma_{j}(z; \mathcal{D})$ that is correct up to an additive error of $\exp\left(-C n^{1/3} (\log n)^{2/3} \ell^2 p^{-2/3}\right).$ Moreover, if the algorithm returns nothing, then $\prod_{i > j} |P(z; x^{(i)})-P(z; x^{(j)})| \le \exp\left(-C n^{1/3} (\log n)^{2/3} \ell^2 p^{-2/3}\right).$
\end{lemma}

\begin{proof}
    By Lemma \ref{CreatingG}, from one sample trace $\tilde{x}$ drawn from $x \sim \mathcal{D},$ for all $0 \le k \le 2 \ell-1$ we can create an unbiased estimator $g_k(\tilde{x}, z)$ for $b_k = \sum a_i P(z; x^{(i)})^k$, which is bounded in magnitude by 
\[(p^{-1} \ell n)^{O(\ell)} \cdot \exp\left(O\left(\ell^2 n/(p^2 L^2)\right)\right) = \exp\left(O\left(n^{1/3} (\log n)^{2/3} \ell^2 p^{-2/3}\right)\right).\]
    Moreover, $g_k(\tilde{x}, z)$ can be computed in $n^{O(\ell)}$ time. We take $\alpha^{-2} \exp\left(C' n^{1/3} (\log n)^{2/3} \ell^2 p^{-2/3}\right)$ samples and compute $\tilde{b}_k$ as the sample mean of the $g_k(\tilde{x}, z)$'s over all sampled traces $\tilde{x}$. By a Chernoff bound argument, we have that with probability at least $1 - \exp\left(-\exp\left(\Theta\left(n^{1/3} (\log n)^{2/3} \ell^2 p^{-2/3}\right)\right)\right),$ if $C'$ is sufficiently large, then $|\tilde{b}_k-b_k| \le \alpha \cdot \exp\left(-20 C n^{1/3} (\log n)^{2/3} \ell^2 p^{-2/3}\right)$ for all $k$.
    
    Now, $\tilde{B}$ will be the matrix with $\tilde{B}_{ij} = \tilde{b}_{i+j-2}$ as before, and we define $\delta := \exp\left(-4 C n^{1/3} (\log n)^{2/3} \ell^2 p^{-2/3}\right).$ Note that $|\tilde{b}_k-b_k| \le \frac{\alpha \cdot \delta}{4 \ell^2},$ so $E := \tilde{B}-V^T A V$ will satisfy $\|E\| \le \|E\|_F \le \frac{\alpha \cdot \delta}{4 \ell}.$ Thus, we can run the algorithm of Lemma \ref{Alg1}. If we return NO, then either $|\det V| = \prod_{i > j} |u_i-u_j| \le \delta,$ or $\sigma_{min}(V^T A V) \le \alpha \cdot \delta.$ In the latter case, $\sigma_{min}(V^T A V) \ge \sigma_{min}(V)^2 \cdot \sigma_{min}(A) \ge \alpha \cdot \sigma_{min}(V)^2,$ so $\sigma_{min}(V) \le \exp\left(-2C n^{1/3} (\log n)^{2/3} \ell^2 p^{-2/3}\right).$ However,
\[\sigma_{min}(V) \ge \frac{|\det V|}{\|V\|^{\ell-1}} \ge \frac{|\det V|}{\|V\|_F^{\ell-1}} = \frac{\prod_{i > j} |u_i-u_j|}{\left(\sum_{i = 1}^{\ell} \sum_{k = 0}^{\ell-1} |u_i|^{2k}\right)^{(\ell-1)/2}} \ge n^{-\Theta(\ell^2)} \cdot \prod_{i > j} |u_i-u_j|,\]
    where the last inequality is true since $|u_i| = |P(z; x^{(i)})| \le n$ when $|z| = 1$. However, we have that $n^{\ell^2} = \exp\left(o\left(n^{1/3} (\log n)^{2/3} \ell^2 p^{-2/3}\right)\right)$, so in either case, $\prod_{i > j} |u_i-u_j| \le \exp\left(-C n^{1/3} (\log n)^{2/3} \ell^2 p^{-2/3}\right)$. Thus, if the subroutine of Lemma \ref{Alg1} returns NO, our algorithm can return nothing.
    
    Otherwise, if our subroutine returns YES, we know that $\sigma_{min}(V^T A V) \ge \frac{\alpha \cdot \delta}{2}$ and $|\det(V)| \ge \frac{\delta}{3},$ so $\sigma_{min}(V) \ge |\det V|/\|V\|_F^{n-1} \ge \delta \cdot n^{-\Theta(\ell^2)}.$ Thus, if we set $\gamma = \delta^2,$ we will have that $\sigma_{min}(V) \ge \gamma \cdot (\max |u_i|)^\ell$ and $\sigma_{min}(V^T A V) \ge \alpha \cdot \gamma.$ Finally, if we set $\eta = \exp\left(-C n^{1/3} (\log n)^{2/3} \ell^2 p^{-2/3}\right),$ we will have $|\tilde{b}_k-b_k| \le \frac{\alpha \cdot \gamma^2 \cdot \eta}{4 \ell^2}$ for all $0 \le k \le 2\ell-1$. Thus, by creating $\tilde{w}$ as done in Lemma \ref{UsingWoodbury}, we will return $\tilde{w}$ so that $\|\tilde{w}-w\|_2 \le \eta = \exp\left(-C n^{1/3} (\log n)^{2/3} \ell^2 p^{-2/3}\right).$ Finally, recalling that $w_i = r_{\ell+1-i} = (-1)^{\ell-i} \sigma_{\ell+1-i}(z; \mathcal{D}),$ if we return $(-1)^{j-1} \tilde{w}_{\ell+1-j}$ as our estimate of $\sigma_j(z; \mathcal{D}),$ then $|(-1)^{j-1} \tilde{w}_{\ell+1-j}-\sigma_j(z; \mathcal{D})| \le \|\tilde{w}-w\|_2 \le \eta = \exp\left(-C n^{1/3} (\log n)^{2/3} \ell^2 p^{-2/3}\right)$ for all $1 \le j \le \ell.$
\end{proof}

Finally, we have the following modification of Lemma \ref{Alg2}, which essentially shows that the above algorithm still works even if with some small probability, we are given ``incorrect'' traces. This will prove useful in converting the proof of Equation \eqref{MainEq3} to a proof of Equation \eqref{MainEq4} in Theorem \ref{MainAlg}.

\begin{proposition} \label{Modification2}
    Let $\mathcal{D}$ be exactly $\ell$-sparse, with notation as in Lemma \ref{Alg2}, and let $\mathcal{D}'$ be any distribution over $\{0, 1\}^n$. Suppose we are given sample access to traces from $x \sim \mathcal{D}'',$ where $\mathcal{D}$ draws a string $x$ from $\mathcal{D}'$ with probability $\kappa$ and a string $x$ from $\mathcal{D}$ with probability $1-\kappa.$ Then, for some sufficiently large constant $C',$ if $\kappa \le \alpha \cdot \exp\left(-C' n^{1/3} (\log n)^{2/3} \ell^2 p^{-2/3}\right)$, the Algorithm of Lemma \ref{Alg2} will still work. In other words, it uses $\alpha^{-2} \cdot \exp\left(O\left(n^{1/3} (\log n)^{2/3} \ell^2 p^{-2/3}\right)\right)$ time and queries and either outputs nothing or, for all $1 \le j \le \ell$, outputs an estimate for $\sigma_{j}(z; \mathcal{D})$ that is correct up to an additive error of $\exp\left(-C n^{1/3} (\log n)^{2/3} \ell^2 p^{-2/3}\right).$ Moreover, if the algorithm returns nothing, then $\prod_{i > j} |P(z; x^{(i)})-P(z; x^{(j)})| \le \exp\left(-C n^{1/3} (\log n)^{2/3} \ell^2 p^{-2/3}\right).$
\end{proposition}

\begin{proof}
    If $C'$ is sufficiently large, then with probability at least $1 - \exp\left(-\exp\left(\Theta\left(n^{1/3} (\log n)^{2/3} \ell^2 p^{-2/3}\right)\right)\right),$ at most $\exp\left(- (C'/2) \cdot n^{1/3} (\log n)^{2/3} \ell^2 p^{-2/3}\right)$ fraction of the traces $\tilde{x}$ that we see come from strings drawn from $\mathcal{D}'$. Now, if we let $\tilde{b}_k'$ denote the sample mean of all $g_k(\tilde{x}, z)$ values over the sampled $\tilde{x}$ drawn from $\mathcal{D},$ we saw that from Lemma \ref{Alg2}, $|\tilde{b}_k'-b_k| \le \alpha \cdot \exp\left(-20 C n^{1/3} (\log n)^{2/3} \ell^2 p^{-2/3}\right)$ for all $k$. Also, if we let $\tilde{b}_k$ denote the sample mean of all $g_k(\tilde{x}, z)$ values over all sampled $\tilde{x}$ drawn from $D'',$ we have $|\tilde{b}_k-\tilde{b}_k'| \le \alpha \cdot \exp\left(-20 C n^{1/3} (\log n)^{2/3} \ell^2 p^{-2/3}\right)$ if $C'$ is sufficiently large. This is because we only draw from $D'$ at most $\exp\left(- (C'/2) \cdot n^{1/3} (\log n)^{2/3} \ell^2 p^{-2/3}\right)$ fraction of the time, and $|g(\tilde{x}, z)| \le \exp\left(O\left(n^{1/3} (\log n)^{2/3} \ell^2 p^{-2/3}\right)\right)$ for all possible $\tilde{x}.$
    
    Thus, 
\[|\tilde{b}_k-b_k| \le 2 \alpha \cdot \exp\left(-20 C n^{1/3} (\log n)^{2/3} \ell^2 p^{-2/3}\right) \le \alpha \cdot \exp\left(-19 C n^{1/3} (\log n)^{2/3} \ell^2 p^{-2/3}\right).\]
    Since the estimates $\tilde{b}_k$ are the only statistics we use in Lemma \ref{Alg2}, the rest of the proof of Lemma \ref{Alg2} goes through.
\end{proof}

    We will ignore the modification of Proposition \ref{Modification2} until it becomes necessary, which will be in proving Equation \eqref{MainEq4} of Theorem \ref{MainAlg}. However, since we will only use Lemma \ref{Alg2} as a black box, we can substitute in Proposition \ref{Modification2} whenever necessary.

\subsection{Determining the Elementary Symmetric Polynomials} \label{SymmetricPoly2}

In Lemma \ref{Alg2} and Proposition \ref{Modification2}, we provide estimates for $\sigma_j(z; \mathcal{D})$, assuming that $\prod |P(z; x^{(i)})-P(z; x^{(j)})|$ is sufficiently large. As a result, we will need a lemma similar to Proposition \ref{Modification}, that shows there is some $z$ in a reasonably sized set such that $\prod |P(z; x^{(i)})-P(z; x^{(j)})|$ is sufficiently large. For the purposes of the linear program that allow us to determine the exact polynomials $\sigma_k(\mathcal{D})$, we will roughly need to show that both $\prod |P(z; x^{(i)})-P(z; x^{(j)})|$ and $\sigma_k(z; \mathcal{D}) - T'(z)$ are sufficiently large for any polynomial $T'(z)$ that could be ``confused'' with $\sigma_k(z; \mathcal{D})$.

\begin{lemma} \label{LowerBoundPoly}
    Let $P$ be a nonzero polynomial of degree at most $\ell^2 \cdot n$ and with all coefficients integers bounded in magnitude by $n^{\ell^2}.$ Also, let $Q$ be a polynomial with real coefficients and degree at most $\ell \cdot n$. Also, suppose that all coefficients of $Q$ are bounded in magnitude by $2n^{2\ell}$ and that $|Q(0)| \ge 1$. Then, for $L = \left\lfloor\left(\frac{n}{\log n \cdot p^2}\right)^{1/3}\right\rfloor,$ there exists some universal constant $C$ and some complex number $z$ with $|z| = 1$, $|\arg z| \le \frac{1}{L}$, and $\arg z$ an integer multiple of $\exp\left(-C n^{1/3} (\log n)^{2/3} \ell^2 p^{-2/3}\right)$ such that
\[|P(z)|, |Q(z)| \ge 3 \exp\left(-C n^{1/3} (\log n)^{2/3} \ell^2 p^{-2/3}\right).\]
\end{lemma}

\begin{proof}
    Throughout the proof, we assume WLOG that $P(0) \neq 0,$ or equivalently, $z \nmid P(z).$ This is allowed because $|P(z)/z| = |P(z)|$ for all $|z| = 1,$ so we can divide $P(z)$ by sufficiently many powers of $z$ until $P(0) \neq 0$.

    First, we prove this result if $\arg z$ does not necessarily have to be an integer multiple of $\exp\left(-C n^{1/3} (\log n)^{2/3} \ell^2 p^{-2/3}\right)$.  Now, let $R(z)$ be the polynomial $P(z) \cdot Q(z).$ Since $P(0)$ is nonzero and integral, and since $|Q(0)| \ge 1,$ we have that $|R(0)| \ge 1.$ Moreover, by our bounds on the degrees and coefficients of both $P$ and $Q$, we have that all coefficients of $R(z)$ are bounded in magnitude by $\ell n \cdot n^{\ell^2} \cdot 2n^{2 \ell} \le 2 n^{4 \ell^2}.$ Finally, letting $S(z) = R(z)/R(0),$ we have that $S(0) = 1,$ all of $S$'s coefficients are bounded in magnitude by $2 n^{4 \ell^2},$ and $|S(z)| \le |R(z)|$ for all $z$, since $|R(0)| \ge 1$.
    
    Now by Theorem \ref{Littlewood}, we have that there is some $z$ with $|z| = 1$ and $|\arg z| \le \frac{1}{L}$ such that
\[|R(z)| \ge |S(z)| \ge \exp\left(-c_1 \cdot L \cdot (1 + \log 2 n^{4 \ell^2})\right) = \exp\left(-c_2 \cdot n^{1/3} (\log n)^{2/3} \ell^2 p^{-2/3}\right)\]
    for some constants $c_1, c_2.$ However, note that by the bound on the degrees and coefficients of $P, Q$, we have that $|P(z)| \le \ell^2 n \cdot n^{\ell^2} \le n^{2 \ell^2}$ and $|Q(z)| \le 2 \ell n \cdot n^{2 \ell} \le n^{4 \ell^2}.$ Therefore, for some constant $c_3 > c_2,$
\[|P(z)|, |Q(z)| \ge \frac{|R(z)|}{\max(|P(z)|, |Q(z)|)} \ge 3 \exp\left(-c_3 n^{1/3} (\log n)^{2/3} \ell^2 p^{-2/3}\right).\]

    Now, we prove the full lemma where $\arg z$ is an integer multiple of $\exp\left(-C n^{1/3} (\log n)^{2/3} \ell^2 p^{-2/3}\right)$ for some $C > c_3.$ To do this, note that by the bound on the degrees and coefficients of $P$ and $Q,$ we have that $\left|\frac{d}{dz} P(z)\right|, \left|\frac{d}{dz} (Q(z))\right| \le (\ell^2 n)^2 \cdot 2n^{2\ell^2} = \exp\left(o\left(n^{1/3} (\log n)^{2/3} \ell^2 p^{-2/3}\right)\right)$ for all $|z| = 1.$ Thus, for some $C > c_3,$ if $|w| = 1$ and $\arg w$ is the closest integer multiple of $\exp\left(-C n^{1/3} (\log n)^{2/3} \ell^2 p^{-2/3}\right)$ to $\arg z$ for the $z$ chosen in the previous paragraph, we will have that 
\[|P(w)|, |Q(w)| \ge 3 \exp\left(-C n^{1/3} (\log n)^{2/3} \ell^2 p^{-2/3}\right). \qedhere\]
\end{proof}

Using the above lemma, we show how to determine the polynomials $\sigma_k(\mathcal{D})$ given sample access to traces from $x \sim \mathcal{D}$.

\begin{lemma} \label{Alg3}
    Suppose we are given sample access to traces from an unknown distribution $\mathcal{D}$ with support size exactly $\ell$, and let all notation be as in Lemma \ref{Alg2}. Also, suppose we know $\alpha, \beta < 1$ such that $\alpha \le \min a_i \le 2 \alpha$ and $\beta \le \prod a_i \le 2 \beta.$ Then, we can determine all coefficients of $\sigma_k(\mathcal{D})$ for all $1 \le k \le \ell$ using $\alpha^{-2} \exp\left(O\left(n^{1/3} (\log n)^{2/3} \ell^2 p^{-2/3}\right)\right)$ queries and time.
\end{lemma}

\begin{proof}
    Let $1 \le k \le \ell$ be fixed. We will attempt to determine the coefficients of $\sigma_k(\mathcal{D})$ one at a time. First, note that $\sigma_k(\mathcal{D})$ is the $k$th symmetric polynomial of $P_1 = P(z; x^{(1)}), \dots, P_\ell = P(z; x^{(\ell)})$, where each polynomial has all coefficients either $0$ or $1$ and has degree at most $n$. Thus, we can write $\sigma_k(z; \mathcal{D}) = T(z) = \sum_{i = 0}^{\ell \cdot n} t_i z^i$, where $t_i$ is a nonnegative integer at most ${\ell \choose k} \cdot n^k \le (2n)^{\ell} \le n^{2 \ell}$.
    
    Let $0 \le i \le \ell \cdot n$ and suppose we know $t_0, \dots, t_{i-1}$ but not $t_i$. Let $T'(z) = \sum_{i = 0}^{\ell \cdot n} t_i' z^i$ be any polynomial with $t_j = t_j'$ for all $1 \le j \le i-1,$ $t_i \neq t_i'$, $0 \le t_j' \le n^{2 \ell}$ for all $j \ge i$, and $t_i'$ is an integer but $t_j'$ for $j > i$ may just be real numbers. Then, if we set $Q(z) := (T(z)-T'(z))/z^i$ and $P(z) = \prod_{i' > j'} (P(z; x^{(i')})-P(z; x^{(j')}))$, it is clear that $P, Q$ satisfy the conditions in Lemma \ref{LowerBoundPoly}. Therefore, for $L = \left\lfloor\left(\frac{n}{\log n \cdot p^2}\right)^{1/3} \right\rfloor,$ there exists some $z$ such that $|z| = 1,$ $|\arg z| \le \frac{1}{L}$, and $\arg z$ an integer multiple of $\exp\left(-C n^{1/3} (\log n)^{2/3} \ell^2 p^{-2/3}\right)$, such that
\[|P(z)|, |T(z)-T'(z)| \ge 3 \exp\left(-C n^{1/3} (\log n)^{2/3} \ell^2 p^{-2/3}\right).\]
    
    Now, we can determine $t_i$ as follows. First, for all $z$ with $|z| = 1,$ $|\arg z| \le \frac{1}{L}$, and $\arg z$ an integer multiple of $\exp\left(-C n^{1/3} (\log n)^{2/3} \ell^2 p^{-2/3}\right)$, we run the algorithm of Lemma \ref{Alg2}, either returning nothing, so $|P(z)| \le \exp\left(-C n^{1/3} (\log n)^{2/3} \ell^2 p^{-2/3}\right)$, or returning some value $h(z)$ such that $|h(z)-\sigma_k(z; \mathcal{D})| = |h(z) - T(z)| \le \exp\left(-C n^{1/3} (\log n)^{2/3} \ell^2 p^{-2/3}\right)$. Let $z_1, \dots, z_R$ be an enumeration of all $z$ such that the algorithm of Lemma \ref{Alg2} returns some $h(z).$ Now, we run the linear program with variables $t_0, \dots, t_{\ell \cdot n}$ and constraints
\begin{alignat*}{3}
\left|\Re (h(z_r))-\sum\limits_{j = 0}^{\ell \cdot n} t_j' \Re(z_r^j)\right| &\le \exp\left(-C n^{1/3} (\log n)^{2/3} \ell^2 p^{-2/3}\right) && \hspace{0.5cm} \forall 1 \le r \le R \\
\left|\Im (h(z_r))-\sum\limits_{j = 0}^{\ell \cdot n} t_j' \Im(z_r^j)\right| &\le \exp\left(-C n^{1/3} (\log n)^{2/3} \ell^2 p^{-2/3}\right) && \hspace{0.5cm} \forall 1 \le r \le R \\
t_j' &= t_j && \hspace{0.5cm} \forall 0 \le j \le i-1 \\
t_i' &\in \{0, 1, \dots, n^{2 \ell}\} \\
0 \le t_j' &\le n^{2 \ell} && \hspace{0.5cm} \forall i+1 \le j \le \ell \cdot n. 
\end{alignat*}
    Note that $t_j' = t_j$ for all $j \le \ell \cdot n$ is a solution, since $|h(z_r)-T(z_r)| \le \exp\left(-C n^{1/3} (\log n)^{2/3} \ell^2 p^{-2/3}\right)$ for all $z_r$. However, if $t_i' \neq t_i$ but $0 \le t_j' \le n^{2 \ell}$ for all $j > i,$ then there exists some $z_r$ such that 
\[\left|\sum_{j = 1}^{\ell \cdot n} t_j' z_r^j - \sum_{j = 1}^{\ell \cdot n} t_j z_r^j\right| = \left|T'(z_r)-T(z_r)\right| \ge 3 \exp\left(-C n^{1/3} (\log n)^{2/3} \ell^2 p^{-2/3}\right).\]
    Thus, by the Triangle Inequality, we have $|T'(z_r)-h(z_r)| \ge 2 \exp\left(-C n^{1/3} (\log n)^{2/3} \ell^2 p^{-2/3}\right),$ so either $|\Re(h(z_r)) - \sum t_j' \Re(z_r^j)| > \exp\left(-C n^{1/3} (\log n)^{2/3} \ell^2 p^{-2/3}\right)$ or $|\Im(h(z_r)) - \sum t_j' \Im(z_r^j)| > \exp\left(-C n^{1/3} (\log n)^{2/3} \ell^2 p^{-2/3}\right)$. Thus, any solution to the linear program must have $t_i' = t_i,$ so we can inductively solve for $t_i$.
    
    Overall, the linear program has $\ell \cdot n + 1$ variables and $O(R+\ell n) = \exp\left(O\left(n^{1/3} (\log n)^{2/3} \ell^2 p^{-2/3}\right)\right)$ constraints. Since only $t_i'$ is integrally constrained, by running $n^{2 \ell}+1$ separate linear programs for all possible values of $t_i',$ one can find a solution to the linear program in $\exp\left(O\left(n^{1/3} (\log n)^{2/3} \ell^2 p^{-2/3}\right)\right)$ time, which gives us $t_i$ assuming we know $t_1, \dots, t_{i-1}.$ Running this for all $1 \le k \le \ell$ and then all $0 \le i \le \ell \cdot n$ completes the proof.
\end{proof}

\subsection{Finishing the Algorithm} \label{FinishAlgorithm}

First, we show that given the elementary symmetric polynomials $\sigma_j(\mathcal{D}),$ we can recover the strings $x^{(1)}, \dots, x^{(\ell)}.$

\begin{lemma} \label{Factor}
    Let $x^{(1)}, \dots, x^{(\ell)}$ be distinct strings in $\{0, 1\}^n$ and let $P_i(z) := P(z; x^{(i)})$ for all $1 \le i \le \ell.$ Suppose we are given $Q_j := \sigma_j(P_1, \dots, P_\ell)$, i.e. the $j$th symmetric polynomial of $P_1, \dots, P_\ell$, for all $1 \le j \le \ell.$ Then, in time polynomial in $n$ and $\ell,$ we can recover the strings $x^{(1)}, \dots, x^{(\ell)}$ in some order.
\end{lemma}

\begin{proof}
    Consider the map $f: \{0, 1\}^n \to \BZ$ sending a string $x \in \{0, 1\}^n$ to $P(2; x) = \sum 2^i x_i.$ Then, note that
\[\prod_{i = 1}^{\ell} \left(z-f(x^{(i)})\right) = z^\ell - Q_1(2) z^{\ell-1} +  \cdots + (-1)^\ell Q_\ell(2).\]
    Since we know the polynomials $Q_1, \dots, Q_\ell,$ we also know the values of $Q_i(2)$ for all $i$, and therefore know the polynomial $\prod (z-f(x^{(i)})).$ Then, if we define $y^{(i)} = f(x^{(i)}),$ note that $y^{(i)} \in \BZ$, $0 \le y^{(i)} \le 2^{n+1},$ and the $y^{(i)}$ values are distinct. It is well known that we can factor $\prod_{i = 1}^{\ell} (z-y^{(i)})$ in time polynomial in $\ell$ and $n$, such as by using the LLL algorithm for factoring polynomials \cite{LLL}, or Rabin's algorithm for finding roots of polynomials modulo $p$ (where we can choose $p$ to be a prime greater than $2^{n+1}$) \cite{Rabin80}. This gives us $y^{(1)}, \dots, y^{(\ell)}$ in some order. But then, by writing out the binary expansion of $y^{(i)}$, we recover the strings $x^{(1)}, \dots, x^{(\ell)}$ in some order.
\end{proof}

As an immediate corollary of Lemma \ref{Alg3} and Lemma \ref{Factor}, we can recover the strings original strings $x^{(1)}, \dots, x^{(\ell)},$ which we now formally state.

\begin{corollary} \label{RecoverStrings}
    Suppose $\mathcal{D}$ is an unknown distribution over $\{0, 1\}^n$ with support size exactly $\ell$, such that $x$ drawn from $\mathcal{D}$ equals $x^{(i)}$ with probability $a_i$ for $a_1, \dots, a_\ell$ positive reals adding to $1$ and $x^{(1)}, \dots, x^{(\ell)}$ distinct strings in $\{0, 1\}^n$. Then, if we are given $0 < \alpha, \beta < 1$ such that $\alpha \le \min a_i \le 2 \alpha$ and $\beta \le \prod a_i \le 2 \beta,$ we can recover the strings $x^{(1)}, \dots, x^{(\ell)}$ in some order, using $\alpha^{-2} \exp\left(O\left(n^{1/3} (\log n)^{2/3} \ell^2 p^{-2/3}\right)\right)$ queries and time.
\end{corollary}

We are now ready to prove the first half of Theorem \ref{MainAlg}. Importantly, we no longer know the values of $\min a_i$ and $\prod a_i$ up to a factor of $2$, but just have some lower bound $\alpha \le \min a_i$. In the second half of Theorem \ref{MainAlg}, we will not even be given any lower bound!

\begin{proof}[Proof of Equation \eqref{MainEq3} of Theorem \ref{MainAlg}]
    Assume WLOG that $\alpha^{-1} \le \min a_i$ is a power of $2$, and let $m \in \BN$ be so that $\alpha = 2^{-m}$. Then, if $\mathcal{D}$ has support size exactly $\ell'$ for some $\ell' \le \ell,$ note that $2^{-m \cdot \ell'} \le (\min a_i)^{\ell'} \le \prod a_i.$ Now, for all $1 \le \ell' \le \ell,$ $1 \le m_1 \le m$, and $1 \le m_2 \le m \cdot \ell'$, we assume that $\mathcal{D}$ has support size exactly $\ell'$, $2^{-m_1} \le \min a_i \le 2^{1-m_1}$, and $2^{-m_2} \le \prod a_i \le 2^{1-m_2}.$ By applying Corollary \ref{RecoverStrings}, for each $\ell', m_1, m_2,$ we get some $\ell'$ strings $x^{(1)}_{\ell', m_1, m_2}, \dots, x^{(\ell')}_{\ell', m_1, m_2}$ using $2^{2 m_1} \cdot \exp\left(O\left(n^{1/3} (\log n)^{2/3} \ell^2 p^{-2/3}\right)\right)$ queries and time.
    
    Thus, by running this over all $\ell', m_1, m_2,$ in $m \cdot 2^{2m} \cdot \exp\left(O\left(n^{1/3} (\log n)^{2/3} \ell^2 p^{-2/3}\right)\right)$ queries and time, we obtain $O(\ell^2 \cdot m^2)$ candidate tuples of strings $\left(x^{(1)}_{\ell', m_1, m_2}, \dots, x^{(\ell')}_{\ell', m_1, m_2}\right)$ representing the support of $\mathcal{D}$. At least one of these tuples will be correct, since there is some $\ell' \le \ell, m_1 \le m, m_2 \le \ell' \cdot m$ such that $\mathcal{D}$ has support exactly $\ell'$, $2^{-m_1} \le \min a_i \le 2^{1-m_1}$, and $2^{-m_2} \le \prod a_i \le 2^{1-m_2}$.
    
    We finish similarly to how we proved Equation \eqref{MainEq1} of Theorem \ref{main}, done at the end of Section \ref{PopulationRecoveryBounds}. Namely, for $L = \left\lfloor\left(\frac{n}{\log n \cdot p^2}\right)^{1/3} \right\rfloor,$ we let $S$ be the set of $z \in \BC$ such that $|z| = 1, |\arg z| \le \frac{2\pi}{L},$ and $\arg z$ is an integer multiple of $n^{-C \cdot L \cdot \ell^2}$ for some constant $C$. Then, using $\eps^{-2} \cdot \exp\left(O\left(n^{1/3} (\log n)^{2/3} \ell^2 p^{-2/3}\right)\right)$ queries, we obtain estimates $h_k(z)$ for $\BE_{x \sim \mathcal{D}} P(z; x)^k$ such that $\left|h_k(z) - \BE_{x \sim \mathcal{D}}P(z; x)^k\right| \le \frac{\eps}{4} n^{-C \cdot L \cdot \ell^2}$ for all $k \le 2\ell-1$ and all $z \in S.$ However, for all $\ell$-sparse distributions $\mathcal{D}'$ with $d_{TV}(\mathcal{D}, \mathcal{D}') \ge \eps,$ there is some $k \le 2\ell-1, z \in S$ such that $\left|\BE_{x \sim \mathcal{D}}P(z; x)^k - \BE_{x \sim \mathcal{D}'}P(z; x)^k\right| \ge \eps \cdot n^{-C \cdot L \cdot \ell^2}$ by Theorem \ref{Distinguish1} (modified by Proposition \ref{Modification}). Therefore, any $\ell$-sparse distribution $\mathcal{D}'$ with 
\[\left|\Re h_k(z) - \BE_{x \sim \mathcal{D}'} (\Re P(z; x)^k)\right| \le \frac{\eps}{4} \cdot n^{-C \cdot L \cdot \ell^2}\]
\[\left|\Im h_k(z) - \BE_{x \sim \mathcal{D}'} (\Im P(z; x)^k)\right| \le \frac{\eps}{4} \cdot n^{-C \cdot L \cdot \ell^2}\]
    for all $z \in S$ and all $k \le 2\ell-1$ must satisfy $d_{TV}(\mathcal{D}, \mathcal{D}') \le \eps.$ By looking at all $O(\ell^2 \cdot m^2)$ candidate tuples of strings $(x^{(1)}, \dots, x^{(\ell')})$ and running a linear program to determine probability values $a_1, \dots, a_{\ell'}$, we will return an appropriate $\mathcal{D}$ with $d_{TV}(\mathcal{D}, \mathcal{D}') \le \eps.$ The linear program takes time 
\[O(\ell^2 \cdot m^2) \cdot \text{poly}(\ell \cdot |S|) \le m^2 \cdot \exp\left(O\left(n^{1/3} (\log n)^{2/3} \ell^2 p^{-2/3}\right)\right).\]
    Thus, writing $m = \log \alpha^{-1},$ the total queries and time are both bounded by $(\alpha^{-2} \log \alpha^{-1} + \eps^{-2}) \cdot \exp\left(O\left(n^{1/3} (\log n)^{2/3} \ell^2 p^{-2/3}\right)\right)$.
\end{proof}

Now, using Proposition \ref{Modification2} instead of Lemma \ref{Alg2}, we can prove Equation \eqref{MainEq4} of Theorem \ref{MainAlg}.

\begin{proof}[Proof of Equation \eqref{MainEq4} of Theorem \ref{MainAlg}]
    Assume that $\mathcal{D}$ has support size exactly $\ell'$ for some $\ell' \le \ell$, with $\BP_{x \sim \mathcal{D}}(x = x^{(i)}) = a_i$ for all $1 \le i \le \ell'.$ Also, assume WLOG that the $x^{(i)}$'s are ordered so that $a_1 \ge a_2 \ge \cdots \ge a_{\ell'}.$ Note that for the constant $C'$ in Proposition \ref{Modification2}, either $a_{\ell'} \ge \eps \cdot \exp\left(-2C' n^{1/3} (\log n)^{2/3} \ell^3 p^{-2/3}\right)$ or there is some $1 \le k \le \ell'-1$ such that $a_{k+1} + \dots + a_{\ell'} \le \frac{1}{4} \cdot \min\left(\eps, a_{k} \cdot \exp\left(-C' n^{1/3} (\log n)^{2/3} \ell^2 p^{-2/3}\right)\right)$. In the former case, define $k := \ell'.$ In the latter case, if we choose $k$ to be as small as possible, then $a_k \ge \eps \cdot \exp\left(-2C' n^{1/3} (\log n)^{2/3} \ell^3 p^{-2/3}\right)$.
    
    Now, let $\mathcal{D}_0$ be the distribution $\mathcal{D}$ conditioned on $x \sim \mathcal{D}$ being $x^{(i)}$ for some $i \le k,$ and let $\mathcal{D}_1$ be the distribution $\mathcal{D}$ conditioned on $x \sim \mathcal{D}$ being $x^{(i)}$ for some $i > k$. Now, letting $\alpha = 2^{-m}$ be so that $\alpha \le a_k < 2 \alpha,$ we draw from $\mathcal{D}_1$ with probability at most $1 - \alpha \cdot \exp\left(-C' n^{1/3} (\log n)^{2/3} \ell^2 p^{-2/3}\right).$ Therefore, assuming we know $\alpha$ and $2^{-m_2} = \beta \le \prod_{i = 1}^{k} a_i \le 2 \beta,$ we can recover the strings $x^{(1)}, \dots, x^{(k)}$ in some order using $\alpha^{-2} \exp\left(O\left(n^{1/3} (\log n)^{2/3} \ell^2 p^{-2/3}\right)\right)$ queries and time.
    
    The rest of the algorithm is nearly identical to the proof of Equation \eqref{MainEq3}: we run over all values of $k \le \ell' \le \ell,$ all values of $m = \log \alpha^{-1} \le \log \eps^{-1} + O(n^{1/3} (\log n)^{2/3} \ell^3 p^{-2/3})$ (using our lower bound for $a_k$), and all values of $m_2 \le k \cdot m$ to get candidates for $x^{(1)}, \dots, x^{(k)}.$ We then run a linear program as in the proof of Equation \eqref{MainEq3} to get a $\ell$-sparse distribution $\mathcal{D}_0'$ such that $d_{TV}(\mathcal{D}_0', \mathcal{D}_0) \le \frac{\eps}{2}.$ Since $a_{k+1} + \dots + a_{\ell'} \le \frac{\eps}{4},$ we have that $\mathcal{D}_0$ is the distribution $\mathcal{D}$ conditioned on an event occurring with probability at least $1 - \frac{\eps}{4},$ so $d_{TV}(\mathcal{D}_0, \mathcal{D}) \le \frac{\eps}{2}.$ Therefore, our final distribution $\mathcal{D}_0'$ satisfies $d_{TV}(\mathcal{D}_0, \mathcal{D}) \le \eps.$ It is clear that the total queries and runtime are both bounded by
\[\eps^{-2} \log \eps \cdot \exp\left(O\left(n^{1/3} (\log n)^{2/3} \ell^3 p^{-2/3}\right)\right),\]
    using the same analysis of the proof of Equation \eqref{MainEq3}.
\end{proof}

\section{Improved Sample Bound for Small $p$} \label{SmallpSection}

In this section, we prove Theorem \ref{Smallp}, thereby getting better sample complexity upper bounds than both Theorem \ref{old} and Theorem \ref{main} when $p \le n^{-1/2}$, assuming that $\eps$ is not too small (i.e. $\eps \ge e^{-\ell^3 \sqrt{n} \log n}$). We note that the proof only requires a few small technical changes to the proof of Theorem \ref{main}, so we will focus on these changes.

First, note that when $p = n^{-1/2},$ by Equation \eqref{MainEq2} of Theorem \ref{main}, we can solve the population recovery problem with $\eps^{-2} \log \eps^{-1} \cdot \exp\left(\Theta\left(\ell^{3} \cdot \sqrt{n} \log n\right)\right)$ samples.

Now, for $p < \frac{1}{2} \cdot n^{1/2},$ consider sampling traces as follows. We will choose some value $t \le n$, and each time we sample a trace $\tilde{x}$, we discard the trace unless the length is at least $t$. If the trace has length at least $t$, let $X = Bin(n, n^{-1/2})|(X \le t),$ i.e. sample $X$ as a random Binomial variable $Bin(n, n^{-1/2})$, conditioned on $X \le t.$ Finally, replace $\tilde{x}$ with $\tilde{x}',$ which is a randomly chosen subsequence of $\tilde{x}$ of length $X$. Then, note that $\tilde{x}'$ has the same distribution as a trace sampled with $p = n^{1/2},$ conditioned on the length of the trace being at most $t$. Now, repeat this process until we have accumulated $\eps^{-2} \log \eps^{-1} \cdot \exp\left(\Theta(\ell^3 \cdot \sqrt{n} \log n)\right)$ samples.

First, we prove the following result.

\begin{lemma} \label{Binomial}
    For any $2\sqrt{n} \le t \le n$ and $0 < p < n^{-1/2},$ $\BP(Bin(n, p) = t) \ge \left(\BP(Bin(n, n^{-1/2}) = t)\right)^{O(\log p^{-1})}.$
\end{lemma}

\begin{proof}
    Note that $\BP(Bin(n, p) = t) = p^t (1-p)^{n-t} {n \choose t}$ whereas $\BP(Bin(n, p) = t) = (\frac{1}{\sqrt{n}})^t (1-\frac{1}{\sqrt{n}})^{n-t} {n \choose t}.$ Note that $p^t \cdot {n \choose t} \ge (pn/(et))^t$ whereas $(\frac{1}{\sqrt{n}})^t \cdot {n \choose t} \le (\sqrt{n}/t)^t.$ Then, $\sqrt{n}/t \le \frac{1}{2}$ but $pn/(et) \ge p/e,$ so 
\begin{equation} \label{a}
p^t {n \choose t} \ge (pn/(et))^t \ge (\sqrt{n}/t)^{O(\log p^{-1})} \ge \left(\left(\frac{1}{\sqrt{n}}\right)^t \cdot {n \choose t}\right)^{O(\log p^{-1})}.
\end{equation}
Finally, we have that $1 \ge 1-p \ge 1-\frac{1}{\sqrt{n}},$ so 
\begin{equation} \label{b}
(1-p)^{n-t} \ge \left(1 - \frac{1}{\sqrt{n}}\right)^{(n-t) \cdot O(\log p^{-1})}.
\end{equation}
    Multiplying Equations \eqref{a} and \eqref{b} finishes the proof.
\end{proof}

We finish as follows. First, choose $t$ as large as possible so that $\BP(Bin(n, n^{-1/2}) \ge t) \ge \eps^{2} (\log \eps^{-1})^{-1} \cdot \exp\left(-\Theta(\ell^3 \sqrt{n} \log n)\right).$ For the constant in the $\Theta$ sufficiently large, we will have that $t \ge 2 \sqrt{n}$. Then, this also means that $\BP(Bin(n, n^{-1/2}) = t) \ge \eps \cdot \exp\left(-\Theta(\ell^3 \sqrt{n} \log n)\right),$ so by Lemma \ref{Binomial}, 
\[\BP(Bin(n, p) \ge t) \ge \BP(Bin(n, p) = t) \ge \eps^{\Theta(\log p^{-1})} \cdot \exp\left(-\Theta(\ell^3 \sqrt{n} \log n \log p^{-1})\right).\]

Now, since we have chosen $t$ as large as possible, the probability that $X > t$ is at most $\eps^2 (\log \eps^{-1})^{-1} \cdot \exp\left(-\Theta(\ell^3 \sqrt{n} \log n)\right),$ which means that the distribution of a trace $\tilde{x}'$ will differ in total variation distance from the distribution of a trace $\tilde{x}$ with $p = n^{-1/2}$ by at most $\eps^{2} (\log \eps^{-1})^{-1} \cdot \exp\left(-\Theta(\ell^3 \sqrt{n} \log n)\right).$ Therefore, even if we accumulate $\eps^{-2} \log \eps^{-1} \cdot \exp\left(\Theta(\ell^3 \cdot \sqrt{n} \log n)\right)$ samples, the total variation distance between the joint distributions of the traces will not differ by more than $0.01,$ since the new samples $\tilde{x}'$ are independent and identically distributed.

Since it takes an expected $\eps^{-\Theta(\log p^{-1})} \cdot \exp\left(\Theta(\ell^3 \sqrt{n} \log n \log p^{-1})\right)$ steps until we get a trace of length $t$, the total number of samples needed is
\[\eps^{-\Theta(\log p^{-1})} \cdot \exp\left(-\Theta(\ell^3 \sqrt{n} \log n \log p^{-1})\right) \cdot \eps^{-2} \log \eps^{-1} \cdot \exp\left(\Theta(\ell^3 \sqrt{n} \log n)\right)\]
\[ = \eps^{-\Theta(\log p^{-1})} \cdot \exp\left(-\Theta(\ell^3 \sqrt{n} \log n \log p^{-1})\right)\]
    samples. Finally, we can simply run the algorithm implied by Equation \eqref{MainEq2} of Theorem \ref{main}.

\section*{Acknowledgments}

I would like to thank Prof. Piotr Indyk for many helpful discussions, as well as for reading and providing helpful feedback on this paper. I would also like to thank Savvy Raghuvanshi for reading a draft of this paper. I would finally like to thank Cameron Musco for informing me about Prony's method and its relation to Proposition 5.1.


\newcommand{\etalchar}[1]{$^{#1}$}

\end{document}